\documentclass[15pt]{article}
\usepackage[utf8]{inputenc}

\usepackage{fullpage}
\usepackage{amssymb,amsmath,amsthm}
\usepackage{algorithm}
\usepackage{algpseudocode}
\usepackage{balance} 
\usepackage{upgreek}
\usepackage[export]{adjustbox}
\usepackage{mathtools}
\usepackage{subfig}
\usepackage{graphicx}
\usepackage{amsmath}
\usepackage{lineno}

\newtheorem{observation}{Observation}
\newtheorem{definition}{Definition}
\newtheorem{theorem}{Theorem}
\newtheorem{openproblem}{Open Problem}
\newtheorem{example}{Example}
\newtheorem{corollary}{Corollary}
\newtheorem{remark}{Remark}

\title{Distance-Preserving Graph Compression Techniques\footnote{This work is funded by the Natural Sciences and Engineering Research Council of Canada (NSERC).}}
\newtheorem{lemma}{Lemma}
\newcommand{\nl}{{n_L}}
\newcommand{\nll}{{n^{2}_L}}
\newcommand{\nrr}{{n^{2}_R}}
\newcommand{\nr}{{n_R}}

\newcommand{\leffs}{{\mathcal{L}}}
\newcommand{\riis}{{\mathcal{R}}}
\newcommand{\leff}{{L}}
\newcommand{\rii}{{R}}
\newcommand{\treeleft}{{T_i^{L}}}
\newcommand{\treeright}{{T_j^{R}}}
\newcommand{\merged}{{e^*}}
\newcommand{\mle}{{\text{MARK\_LEFT}}}
\newcommand{\umle}{{\text{UNMARK\_LEFT}}}
\newcommand{\mr}{{\text{MARK\_RIGHT}}}
\newcommand{\unmr}{{\text{UNMARK\_RIGHT}}}
\newcommand{\sumleft}{S_L}
\newcommand{\sumright}{S_R}
\newcommand{\sumleftu}{S_{LU}}
\newcommand{\sumleftm}{S_{LM}}
\newcommand{\sumrightu}{S_{RU}}
\newcommand{\sumrightm}{S_{RM}}
\newcommand{\eij}{E^{(v_i, u_j)}}

\newcommand{\nw}{\mathcal{W}}
\newcommand{\optimalmarking}{M^*}
\newcommand{\markingleft}{M_L}
\newcommand{\markingright}{M_R}

\let\oldequation\equation
\let\oldendequation\endequation

\renewenvironment{equation}
  {\linenomathNonumbers\oldequation}
  {\oldendequation\endlinenomath}

\author{Amirali Madani$^1$ \and Anil Maheshwari$^1$}
\date{%
    $^1$ Carleton University\\%
}
\begin{document}
\maketitle
\begin{abstract}

We study the problem of distance-preserving graph compression 
for weighted paths and trees. The problem entails a weighted 
graph $G = (V, E)$ with non-negative weights, and a subset of edges 
$E^{\prime} \subset E$ which needs to be removed from G (with 
their endpoints merged as a supernode). The goal is to 
redistribute the weights of the deleted edges in a way that 
minimizes the error. The error is defined as the sum of the 
absolute differences of the shortest path lengths between different pairs of nodes before and after contracting $E^{\prime}$. 
Based on this error function, we propose optimal approaches for 
merging any subset of edges in a path and a single edge in a 
tree. Previous works on graph compression techniques aimed at 
preserving different graph properties (such as the chromatic 
number) or solely focused on identifying the optimal set of 
edges to contract. However, our focus in this paper is on 
achieving optimal edge contraction (when the contracted edges are provided 
as input) specifically for weighted trees and paths. 
    
\end{abstract}
\section{Introduction and Related Work}

Graphs have become increasingly relevant for solving real-world 
problems, leveraging their numerous characteristics 
\cite{application1, application2}. However, many of these graphs 
are incredibly large, consisting of trillions of edges and 
vertices, which poses scalability challenges for modern systems 
\cite{scalability, scalability2, scalability3}.  Consequently, 
graph compression techniques have garnered significant research 
interest in recent years, aiming to obtain a smaller graph while 
retaining the essential properties of the original input graph. 
Different names, such as graph compression \cite{shrink}, graph 
summarization \cite{summarization}, graph modification 
\cite{modification}, and graph contraction \cite{contraction}, 
have been used in the literature to describe this problem, each 
within its specific context, leading to various proposed 
approaches. Regardless of the terminology or context, most of 
these problems focus on reducing the size of the graph while 
preserving a specific property \cite{shrink}, while some 
approaches aim to modify a graph to satisfy a given property 
\cite{nphard, nphard2}. Furthermore, graphs can be compressed in 
different ways, such as vertex deletions, edge deletions, and 
edge contractions. It is worth noting that many of these 
resulting graph modification problems are NP-hard, as indicated 
in \cite{nphard}.\\\\
A very relevant problem is commonly referred to in the 
literature as the \textit{blocker} problem \cite{contraction}. 
Given a graph $G$, integers $k$ and $d$, an invariant $\pi: 
\mathcal{G} \xrightarrow{} \mathbb{R}$, and some modification 
operations (such as edge contractions), a blocker problem asks 
whether there exists a set of at most $k$ graph modification 
operators such that in the resulting graph $G^{'}$, 
$\pi\left(G^{\prime}\right) \leq \pi(G)-d$ holds. In recent 
years, blocker problems have been studied for various graph 
properties such as the chromatic number 
\cite{chromatic,indset3}, maximum weight independent set and 
minimum weight vertex cover \cite{indset}, maximum independent 
set \cite{indset2, indset3}, the clique number \cite{indset3}, 
the total domination number \cite{dominate}, diameter 
\cite{diameter}, and maximum weight clique \cite{maximumclique}. 
A lot of these blocker problems are defined as 
\textit{contraction} problems, in which graphs can only be 
modified via edge contractions. More precisely, given a graph 
$G$, integers $k$ and $d$, and an invariant $\pi: \mathcal{G} 
\xrightarrow{} \mathbb{R}$, $\operatorname{CONTRACTION}(\pi)$ 
asks whether there exists a set of at most $k$ edge contractions 
that results in a graph $G^{\prime}$ with 
$\pi\left(G^{\prime}\right) \leq \pi(G)-d$. Galby \textit{et 
al.} \cite{galby2021reducing} studied the contraction problems 
in which a specific edge could be provided as input (in addition 
to $\pi$). As an important contribution, Galby \textit{et al.} 
\cite{galby2021reducing} proved that, unless $\textbf{P=NP}$, 
there exists no polynomial-time algorithm that decides whether 
contracting a given edge reduces the total domination number. 
Biedl \textit{et al.} \cite{biedl} studied the problem of flow-
preserving graph simplification, which is the problem of finding 
a set of edges whose removal does not change the maximum flow of 
the underlying network.\\\\
Shortest path queries are crucial to various domains, including 
search engines \cite{kargar2011keyword}, networks 
\cite{networks, networks2} and transportation 
\cite{transportation, transportation2}. In a more relevant work 
to this paper, Bernstein \textit{et al.} \cite{bernstein} 
studied a slightly different variant of 
$\operatorname{CONTRACTION}$. In their work, Bernstein 
\textit{et al.} \cite{bernstein} focused on compressing a given 
graph as much as possible, while permitting only a limited 
amount of distance distortion among any pair of vertices. Given 
a tolerance function $\varphi(x)=x / \alpha-\beta$, with $\alpha 
\geq 1$ and $\beta \geq 0$, Bernstein \textit{et al.} 
\cite{bernstein} studied the problem of finding the maximum 
cardinality set of edges whose contraction results in a graph 
$G^{\prime}$ such that $\operatorname{d}_{G^{\prime}}(u, v) \geq 
\varphi\left(\operatorname{d}_{G}(u, v)\right)$ for all $ u,v 
\in G$. However, in their work, they only focused on 
\textit{finding} an optimal set instead of optimally 
\textit{redistributing} the weights. More specifically, after 
finding the optimal set $E^{\prime}$, they set the weight of 
each edge $e \in E^{\prime}$ to zero.\\\\
Unlike the work by Bernstein \textit{et al.} \cite{bernstein}, 
the work by Sadri \textit{et al.} \cite{shrink} focuses on 
optimally redistributing the weights. However, they do not 
provide any bound guarantees on the amount of error. Precisely, 
they assess the efficiency of their proposed approach by a set 
of experimental studies. Moreover, their weight redistribution 
approach for trees ignores the size of each subtree rooted at 
the endpoints of a given contracted edge, which is a key factor 
in deciding the optimal assignment as we will show in this 
paper. More recently, Liang \textit{et al.} \cite{reachability} 
studied the problem of reachability-preserving graph compression 
techniques. There have also been other works related to graph 
compression for unweighted and weighted graphs, as listed in 
\cite{unweighted, unweighted2}. Zhou \textit{et al.} 
\cite{zhou2010network} proposed an efficient approach to remove 
a large portion of the edges in a network without affecting the 
overall connectivity by much. Ruan \textit{et al.} 
\cite{unweighted} studied the minimum gate-vertex set discovery
(MGS) problem. The MGS problem is concerned with finding the 
minimum cardinality set of vertices, designating them as gate 
vertices, using which every non-local pair of vertices (whose 
distance is above some threshold) is able to recover its 
distance in the original network. However, the work by Ruan et al. \cite{unweighted} only studies unweighted graphs. \\\\
\textbf{Where our work stands in the literature: }Surprisingly, 
there has been little attention in the literature to one 
particular side (discussed momentarily) of the distance-preserving graph compression problem. To the best of our 
knowledge, all existing works have either only focused on 
finding an optimal set of edges to contract or have not provided 
any bounds on the amount of error. We study a different problem: 
instead of choosing \textit{which} optimal edge to contract, we 
are interested in finding out \textit{how} to contract a given 
edge optimally. Even though we still study distance-preserving 
graph compression, our focus is mainly on optimally modifying 
the graph after a given edge has been contracted. Our primary 
modification operation is changing the edge weights of the 
graph.  It is worth noting that this problem has received 
limited attention in the literature, with the closest existing 
work being the study by Sadri \textit{et al.} \cite{shrink}. 
Their approach involves solving a system of equations to 
determine the new edge weights in the resulting graph. However, 
their analysis of the problem has certain limitations. Firstly, 
they do not offer any optimal guarantees for their weight 
distribution technique. Furthermore, their weight redistribution 
method does not account for the sizes of the individual 
subgraphs connected to a given edge. In contrast, as we will 
demonstrate throughout this paper, the size of a subtree 
(particularly in the context of paths and trees) plays a crucial 
role in achieving optimal weight redistribution.\\\\
\textbf{The organization of the paper:} The remainder of this 
paper is organized as follows. In Section \ref{summary}, we 
present a summary of our main results along with some comments 
and details regarding each contribution. In Section \ref{avalin}, 
we describe the notation used in the paper, using which we 
formally define the scope of our paper in Section \ref{dovomin}. 
In Section \ref{paths}, we study the problem of distance-preserving graph compression for weighted paths, where we prove 
optimal approaches to contracting any set of $k$ edges. In 
Section \ref{treescompressed}, we study the problem of graph 
compression for weighted trees, where we provide an optimal 
linear-time algorithm for contracting a single edge. We present 
the concluding remarks of this paper along with some potential 
avenues of future work in Section \ref{conclusion}.

\subsection{Contributions and Results}
\label{summary}
 In Section \ref{paths}, we study the problem of distance-preserving graph compression for weighted paths. 
\begin{itemize}
    \item As a warm-up, we prove an optimal bound for merging\footnote{As defined in Section \ref{prel}, we use the terms \textit{merging} and \textit{contracting} interchangeably. They both refer to the act of contracting an edge or a set of edges.} a 
    single edge in a path topology in Section \ref{oneedgepath}, 
    whose main result is stated in Theorem \ref{l3aval}.
\end{itemize}
In Section \ref{oneedgepath}, we present a method for 
transforming any weight redistribution for a given merged edge 
$e^{*}$ to another redistribution in which only the weights of 
its neighbouring edges are altered. \\
\begin{itemize} 
    \item We present Algorithm \ref{pathsalg} for merging any 
    set of $k\leq  \frac{n}{2}$ independent edges (edges that 
    have no endpoints in common and induce a matching on the 
    path) in a path of size $n$.
\end{itemize}

We note that Algorithm \ref{pathsalg} produces suboptimal 
solutions when applied to a contiguous subpath (a connected 
subgraph) of the given input path. We relate this suboptimal 
performance to the distinction between merging two regular 
vertices and two \textit{supernodes}. We thoroughly investigate 
this distinction in Lemma \ref{general}, where we present an 
optimal redistribution for merging two supernodes.  \\
\begin{itemize}
    \item Having the suboptimal performance of Algorithm 
    \ref{pathsalg} for merging subpaths in mind, in Section 
    \ref{subp} we study the problem of finding the optimal 
    redistribution for any connected subgraph of a given input 
    path. The optimal method for contracting any contiguous 
    subpath of the input path is presented in Theorem 
    \ref{subp2}.
\end{itemize}
In Section \ref{treescompressed}, we study the problem of 
distance-preserving graph compression for the tree topology 
where we present optimal approaches for merging a single edge in 
a weighted tree. To this end, we define a relevant problem, 
which we refer to as the \textit{marking problem}.  The 
objective of the marking problem is to minimize the error, as 
defined in Section \ref{prel}, by marking a subset of the 
neighbouring edges of the merged edge $e^{*}$ (with weight 
$w^{*}$). For merging an edge $e^*$ with weight $w^*$, an edge 
$e_i$ is said to be marked if its new weight $w^{\prime}
(e_i)=w(e_i)+ w^*$. As a warm-up, in Section \ref{equal}, we 
study the marking problem for a tree in which the neighbouring 
subtrees of $e^{*}$ are of equal sizes. For such edges, we show 
that the optimal marking is achieved when all edges to the left 
or right of (but not both) $e^{*}$ are marked. In Section 
\ref{varying}, we generalize the findings of Section \ref{equal} 
and present an optimal marking for any merged edge $e^{*}$ in a 
weighted tree.

The definition of the marking problem implies that an edge can 
either be fully marked or unmarked. It is non-trivial to see 
whether \textit{fractionally marking} the edges produces better 
results. Therefore, in Section \ref{fract}, we thoroughly 
investigate the distinction between the marking problem 
(Definition \ref{marking}) and the fractional marking problem 
(Definition \ref{fracdef}) and conclude that any solution to the 
latter  can be transformed into another solution to the former 
without worsening the error value. 
\begin{itemize}
    \item We present an $\mathcal{O}(|V|)$-time algorithm for 
    finding an optimal marking for $e^{*}$ in Algorithm 
    \ref{treesalg2}.
\end{itemize}

\section{Preliminaries}
In this section, we first discuss the common notation (Section 
\ref{avalin}) and then present some additional definitions 
(Section \ref{dovomin}) that help describe the scope of our paper. 
In Section \ref{lemmaa}, we present a simple number-theoretic 
lemma that is later used in some proofs of the paper. Throughout 
this section, we use the path in Figure \ref{runningexample} as 
the running example of the definitions.
\label{prel}

\begin{figure*}[h]
    \centering
        \subfloat[]{{\includegraphics[width=350pt]{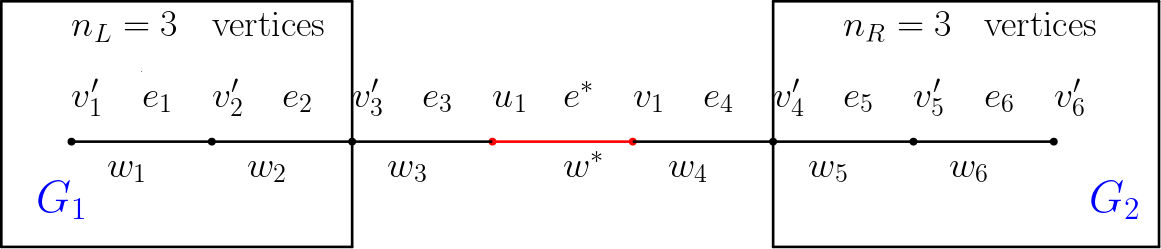}}}
        \hspace{2 cm}
        \subfloat[]{{\includegraphics[width=300pt]{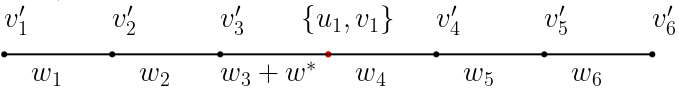}}}
        \hspace{2 cm}
    \caption{The path used as the running example in Section 
    \ref{prel}: (a) A path of 8 vertices, with regular edges 
    denoted by $e_i, w_i=w(e_i)$, 
    and the contracted edge highlighted in red and denoted by 
    $e^*=(u_1,v_1), w^*=w(e^*)$, and (b) The same path after 
    contracting $e^*$ and marking $e_3$ by setting $w^{\prime}
    (e_3)=w_3+w^*$. In this example, $n_L=n_R$, this is not 
    always the case.}
    \label{runningexample}
\end{figure*}
\subsection{Notation} 
\label{avalin}
Let $G=(V,E)$ denote a graph with $V$ and $E$ as its sets of 
vertices and edges respectively. With every edge $e \in E$, we 
associate a weight $w(e)$ $w: E \xrightarrow[]{} 
\mathbb{R}_{\geq 0}$. We sometimes denote an edge $e$ by 
$(u,v)$, where $u,v \in V$ are referred to as the 
\textit{endpoints} of $e$. Throughout this paper, we frequently 
denote edges and vertices using subscripts (for instance $e_i$ 
and $v_i$) and superscripts for merged edges (for instance 
$e^*$). When the context is clear, we sometimes abuse the 
weight notation and denote the weight of $e_i$ and $e^*$ by 
$w_i$ and $w^*$ respectively. We denote the number of vertices 
($|V|$) by $n$ and a path of $n$ vertices by $P_n$. Throughout 
this paper, we frequently use $n_L$ and $n_R$ in different 
contexts to denote different quantities. However, in most cases, 
we denote by $n_L$ and $n_R$ the number of vertices to the left 
and right of a given vertex of a path respectively (including 
itself). For instance, in Figure \ref{runningexample}-(a), $n_L$ 
and $n_R$ denote the number of vertices to the left of (and 
including) $v^{\prime}_3$ and the right of (and including) 
$v^{\prime}_4$, respectively. Formally, let $G_1$ be one of the 
two connected components of $H=G-\{ e_3=(v^{\prime}_3,u_1), e^*=
(u_1,v_1), e_4=(v_1,v^{\prime}_4)\}$ that is adjacent to 
$v^{\prime}_3$, and let $G_2$ be the connected component of $H$ 
that is adjacent to $v^{\prime}_4$. We have 
\begin{linenomath*}$$n_L=\left|\{v| v \in G_1 \}\right|, n_R=\left|\{v| v \in G_2 \}\right|$$\end{linenomath*}
For instance, in Figure \ref{runningexample}-(a), $n_L=3$ 
because $G_1$ includes vertices $\{v^{\prime}_1, v^{\prime}_2, 
v^{\prime}_3\}$, and $n_R=3$ because $G_2$ includes vertices $\{v^{\prime}_4,v^{\prime}_5, v^{\prime}_6\}$. Therefore, in this 
paper, we assume that the graph is laid out in the plane and the 
edge to be merged ($e^*$ in Figure \ref{runningexample}-(a)) is 
horizontal. This assumption will simplify the description of our 
results.
\subsection{Additional Definitions}
\label{dovomin}
We now provide some additional definitions for defining the scope of our paper. 
\begin{definition}
For a weighted graph $G=(V,E)$, the \textit{distance} between 
two vertices $u, v \in V,$ denoted by $d_G(u,v)$, is the length 
of the shortest weighted path between $u$ and $v$ in $G$.
\end{definition}

\begin{definition}
A \textit{merged} edge, or a \textit{contracted} edge, is one 
whose endpoints are merged, and the edge itself is removed from 
the graph.     
\end{definition}
For instance, $e^*=(u_1,v_1)$ in Figure \ref{runningexample}-(a) 
(highlighted in red) is a contracted edge. After contracting 
$e^*,$ the path of Figure \ref{runningexample}-(a) is 
transformed into the one in Figure \ref{runningexample}-(b). 
\begin{definition}
    A \textit{supernode} is a node containing a subset 
    $V^{\prime} \subset V$ of the nodes in the original graph, 
    which is a result of a series of edge contractions. We 
    denote the set of all supernodes by $V_s$.
    \label{superr}
\end{definition}
\begin{definition}
    For a supernode $v\in V_s$, the \textit{cardinality} of $v$, 
    denoted by $\mathcal{C}, \; \mathcal{C}: V_s \rightarrow 
    \mathbb{N},$ is the number of regular vertices it contains.
    \label{superrr}
\end{definition}
In the path of Figure \ref{runningexample}-(b), $\{u_1, v_1\}$ 
is a supernode with cardinality 2. 
\begin{definition}
    Let $G=(V,E)$ be a graph with weight function $w: E 
    \rightarrow \mathbb{R}_{\geq 0}$, and let $e^* \in E$ be the 
    merged edge. A \textit{weight redistribution} is a new 
    weight function $w^{\prime}: E \rightarrow \mathbb{R}_{\geq 
    0}$ in which $w^{\prime}(e_i)=w(e_i) + \epsilon_i, \; 
    \forall e_i \in E, \; \epsilon_i \in \mathbb{R} $.
\end{definition}
In the path of Figure \ref{runningexample}-(b), the weight 
redistribution sets the edge weights of Figure \ref{runningexample}-(a) as $w^{\prime}(e)=w(e)+w(e^*)$ if $e=e_3$, and $w^{\prime}
(e)=w(e)$ otherwise.

\begin{definition}
        With reference to a given merged edge $e^*$ in a graph $G=(V,E)$ with the associated weight function $w: E \rightarrow \mathbb{R}_{\geq 0}$ and a new \textit{weight redistribution} function $w^{\prime}: E \rightarrow \mathbb{R}_{\geq 0}$, an edge $e_i$ is said to be \textit{marked} if $w^{\prime}(e_i)=w(e_i)+w(e^{*})$, \textit{unmarked} if $w^{\prime}(e_i)=w(e_i)$, and \textit{altered} otherwise.
\end{definition}
As shown in Figure \ref{runningexample}-(b), $e_3$ is marked and all other edges are unmarked.
\begin{definition}
    With reference to a set of merged edges $E_m \subset E$, the \textit{set of merged vertices} $V_m$ consists of all vertices with at least one endpoint in $E_m$, or  $V_m=\{v| v,u \in V, \exists e=(u,v) \in E_m\}$. The set of \textit{unmerged vertices} is defined as $\overline{V_{m}}=V- V_m$.
\end{definition}
In the path of Figure \ref{runningexample}, we have $E_m=\{e^*\}$, $V_m=\{u_1,v_1\}$, and $\overline{V_m}=\{v^{\prime}_1, v^{\prime}_2, v^{\prime}_3, v^{\prime}_4, v^{\prime}_5, v^{\prime}_6\}$.
   With reference to a set of merged edges $E_m \subset E$ and a weight redistribution $w^{\prime}$, let $G^{'}$ be the resulting graph after contracting the edges in $E_m$ and setting the new edge weights according to $w^{\prime}$. The error associated with $w^{\prime}$ with respect to $E_m$ is denoted by $|\Delta E|$ and calculated as: 
\begin{linenomath*}\begin{equation}
\label{hehe}
|\Delta E|=\sum_{u \in V_m, v \in \overline{V_m} , \text{ or } u, v \in \overline{V_m}, u\neq v}\left|d_{G}( u, v)-d_{ G^{\prime} }(u, v)\right|
\end{equation}\end{linenomath*}
In other words, the error is equal to the sum of the absolute differences (between $G$ and $G^{\prime}$) of all shortest path lengths between vertices $u,v$, at least one of which is in $\overline{V_m}$.\\
Returning to our example in Figure \ref{runningexample}, the error function of Eq.~(\ref{hehe}) sums up the absolute values of the shortest path differences among the vertices of $\overline{V_m}=\{v^{\prime}_1, v^{\prime}_2, v^{\prime}_3, v^{\prime}_4, v^{\prime}_5, v^{\prime}_6\}$, and between the vertices of $\overline{V_m}$ and the vertices of $V_m=\{u_1, v_1\}$. As the final example, we now explain how the distance difference between one of the aforementioned pairs of vertices is calculated. In Figure \ref{runningexample}, the shortest path value difference between $v^{\prime}_1$ and $u_1$ changes from $w_1+w_2+w_3$ in G (Figure \ref{runningexample}-(a)) to $w_1+w_2+w_3+w^*$ in $G^{\prime} $ (Figure \ref{runningexample}-(b)). The error induced by this change is thus equal to $|w_1+w_2+w_3+w^*-w_1-w_2-w_3|=w^*$.

 We are now ready to present the formal definition of our first studied problem:
\begin{definition}
    \textbf{Distance-Preserving Graph Compression}: Given a graph $G$, and a set of contracted edges $E_m$, the problem of distance-preserving graph compression is to find a weight redistribution $w^{\prime}$ for which $|\Delta E|$ is minimized.
\end{definition}
\subsection{A Number-Theoretic Lemma}
\label{lemmaa}
The following lemma is used in some of the proofs of this paper:
\begin{lemma}
For all real numbers $A, B, C, D, x, y $, let $\alpha_1 = |x-A|+|x-A-B|$ and $\alpha_2=|y-C|+|y-B-C|$. We have $\alpha_1 \geq B$ and $\alpha_2 \geq B$. Furthermore, $\alpha_1=B$, $\alpha_2=B$ for $A \leq x \leq A+B$ and $C \leq y \leq B+C$.
\label{l12}
\end{lemma}
\begin{proof}
We prove this using contradiction. Let us prove $\alpha_1 \geq B$, the other proof will be analogous. For the sake of contradiction, assume that $\alpha_ 1< B$. We have four cases depending on whether the values inside the absolute value function ($x-A$ and $x-A-B$) are positive or negative. Note that $|a| = a$ when $a\geq 0$, and $|a|=-a$ otherwise.
\begin{itemize}
    \item \textbf{Case 1:} $x < A$ and $x<A+B$:
\begin{linenomath*}    $$\alpha_1= A-x + A+B -x < B \xrightarrow{} A < x$$\end{linenomath*}
    which contradicts the assumption ($x<A$).
    \item \textbf{Case 2:} $x<A$ and $x \geq A+B$: These two conditions imply that $B<0$, we have:
    \begin{linenomath*}$$\alpha_1 = A-x + x-A-B <B \xrightarrow{} 0 <2B $$\end{linenomath*}
    which is a contradiction since $B<0$. 
    \item \textbf{Case 3:} $x\geq A $ and $x<A+B$:
    \begin{linenomath*}$$\alpha_1 = x-A +A+B-x <B \xrightarrow{} 0 <0$$\end{linenomath*}
    which is impossible.
    \item \textbf{Case 4:} $x\geq A$ and $x\geq A+B$:
    \begin{linenomath*}$$x-A+x-A-B<B\xrightarrow[]{}x<A+B$$\end{linenomath*}
    which contradicts the assumption.
\end{itemize}
Since we get a contradiction for every possible case, we have $\alpha_1\geq B$ and $\alpha_1=B$ for $A \leq x \leq A+B$. Similarly, we have $\alpha_2\geq B$ and $\alpha_2= B$ for $C \leq x \leq B+C$.
\end{proof}
We will also use the following corollary.
\begin{corollary}
Given two real numbers $x, z$ we have $|z|-|x| \leq |z-x|$. 
\label{obs}
\end{corollary}
\begin{proof}
    Using Lemma \ref{l12}, we have the following for two real numbers $x$ and $z$:
\begin{linenomath*}    $$|z| \leq |x|+|z-x|\xrightarrow[]{} |z|-|x|\leq |z-x|$$\end{linenomath*}
\end{proof}
\section{Graph Compression for Paths}
\label{paths}

    \begin{figure*}[h]
    \centering
        \subfloat[]{{\includegraphics[width=200pt]{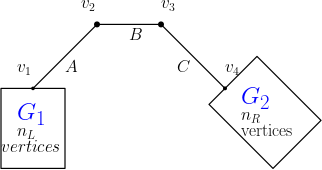}}}
        \subfloat[]{{\includegraphics[width=200pt]{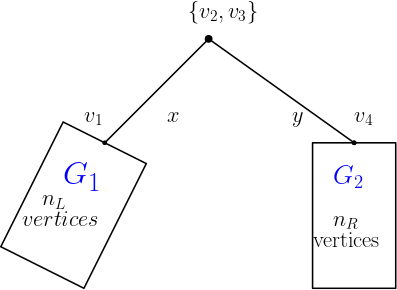}}}\

    \caption{Merging a single edge $e^*=(v_2,v_3)$ with weight $B$ in a path of $n$ vertices, with $n_L \geq 0$, $n_R \geq 0$ vertices on the left of $v_1$ and the right of $v_4$ respectively ($n_L +n_R=n-2$) (a) The original graph before merging $v_2$ and $v_3$ into a supernode. The neighbouring edges of $e^*$ have weights $B$ and $C$.  (b) The modified path after merging $v_2$ and $v_3$ into a supernode. }
    \label{france}
\end{figure*}

In this section, we study the problem of distance-preserving graph compression for a weighted path with non-negative weights. The paths in this section all have $n \geq 3$ vertices since the compression problem for a two-vertex path is trivial.

The remainder of this section is organized as follows. As a warm-up, we provide optimal bounds for merging a single edge in Section \ref{oneedgepath}. In Section \ref{super}, we study the path compression problem for an edge connecting two supernodes, each consisting of a subset of nodes from the path. Two generalizations of the results of Section \ref{oneedgepath} for contracting any subpath (a contiguous subpath of the original graph) and any set of independent edges (that induce a matching in the original path) are provided in Section \ref{subp} and Section \ref{indp} respectively.
\subsection{A Tight Lower Bound for Merging One Edge}
\label{oneedgepath}
This section presents a tight lower bound on the optimal error (Eq. (\ref{hehe})) associated with merging a single edge in a path topology.

 As seen in Figure \ref{france}, the edge between $v_2$ and $v_3$ is merged, and only the immediate edge weights are altered to $x$ and $y$. Later in this section (Lemma \ref{construction}), we show why it is sufficient to alter only the immediate edge weights ($A$ and $B$ in Figure \ref{france}) to get the minimum amount of error. Note that, for merging a single edge (Figure \ref{france}) we have: 
\begin{equation}
    n_L + n_R =n-2= \left|\overline{V_m} \right|
    \label{nlr1}
\end{equation}

The following theorem is now presented:
\begin{theorem}
Let $|\Delta E|$ be the error associated with merging a single edge $e^*=(v_2, v_3)$ (with weight $B$) in a path $P_n, n \geq 3$ (Figure \ref{france}). Furthermore, let $V_m=\{v_2, v_3\}$ and $\overline{V_m}= V- V_m$. We have $|\Delta E| \geq (n-2) B=|\overline{V_m}|B$. Moreover, this lower bound is tight and can be achieved by marking the left neighbour of the merged edge. If the merged edge has no left or right neighbour, the lower bound can be achieved by simply contracting the edge, and no further modifications (weight changes) are required. 
\label{l3aval}
\end{theorem}
\begin{proof}
Figure \ref{france} depicts the situation in which edge $e^*$ with weight $B$ is merged. We first assume that $e^*$ has a left neighbour, and we handle the no-neighbour exception at the end of the proof. As seen in Figure \ref{france}-(b), let $x$ and $y$ denote the new edge weights of the neighbouring edges of $e^*$, and let $G_1$ (with $n_L$ vertices) and $G_2$ (with $n_R$ vertices) denote the subpaths rooted at $v_1$ and $v_4$ respectively. We denote the error by $|\Delta E|$ and classify it into different parts (in accordance with Eq. (\ref{hehe})):
\begin{enumerate}
    \item The error between two vertices $u \in G_1, v\in G_2$ is $|x+y-A-B-C|$. The only affected portion of such a shortest path is the subpath between $v_1$ and $v_4$, the value of which changes from $A+B+C$ (in Figure \ref{france}-(a)) to $x+y$ (in Figure \ref{france}-(b)). Summing over all such pairs $u \in G_1,v \in G_2$, the total amount of error is $n_L n_R |x+y-A-B-C|$.
    \item Between two vertices $u,v \in G_1$, there is no error, because the shortest path value between all such pairs of vertices is unchanged. Similarly, between two vertices $u,v \in G_2$, there exists no error. 
    \item The error between a vertex $u \in G_1$ and the vertices in $V_m=\{v_2, v_3\}$ is $|x-A|+ |x-A-B|$. In the path from $u$ to $v_2$, the only changed (with reference to edge weights) subpath is the subpath between $v_2$ and $v_1$ which changes from $A$ (Figure \ref{france}-(a)) to $x$ (Figure \ref{france}-(b)), inducing an error of $|x-A|$. Similarly, the error between some vertex $u \in G_1$ and $v_3$ is $|x-A-B|$ as that is the amount by which the weight of the subpath from $v_1$ to $v_3$ changes. The total amount of error between all vertices $u \in G_1$ and the vertices in $V_m=\{v_2, v_3\}$ is therefore $n_L (|x-A|+|x-A-B|)$.
    \item By similar reasoning to the one provided above, the total amount of error between all vertices $u \in G_2$ and the vertices in $V_m=\{v_2, v_3\}$ is equal to $n_R(|y-C|+|y-B-C|)$.
\end{enumerate}
Therefore, we can formulate $|\Delta E|$ as
\begin{linenomath*}
$$|\Delta E|=n_L (|x-A|+|x-A-B|) + n_R (|y-C|+|y-B-C|) +n_Ln_R |x+y-A-B-C|= n_L \alpha_1 +n_R \alpha_2  + n_L n_R |x+y-A-B-C|$$
\end{linenomath*}
where $\alpha_1$ and $\alpha_2$ are the values defined in Lemma \ref{l12}. Using Lemma \ref{l12}, we know $\alpha_1 \geq B$ and $\alpha_2 \geq B$: 
\begin{linenomath*}$$|\Delta E| \geq B(n_L +n_R)+  n_L n_R |x+y-A-B-C|$$\end{linenomath*}
Using Eq. (\ref{nlr1}):
\begin{linenomath*}$$|\Delta E| \geq B(n-2)+  n_L n_R |x+y-A-B-C| \geq B(n-2)= |\overline{V_m}| B$$\end{linenomath*}
Which proves the first part of the theorem. As for the second part, we now show that this lower bound is tight. By marking the left neighbouring edge of $e^*$ (effectively setting $x=A+B$ and $y=C$) we get:
\begin{linenomath*}\begin{align*}
    |\Delta E|=& n_L (|x-A|+|x-A-B|) + n_R (|y-C|+|y-B-C|) +n_Ln_R |x+y-A-B-C|\\
    \xrightarrow{\text{setting } x=A+B, y=C} =& (n_L+n_R)B= (n-2)B =|\overline{V_m}|B
\end{align*}\end{linenomath*}
This analysis concludes the proof for the case where $e^*$ has a left neighbour.

If $e^*$ has no left neighbour, $n_L=0$ and no shortest path crossing $e^*$ is affected. For each shortest path starting from $v_4$ (and its right-side vertices) and terminating at $v_2$ and $v_3$ there is an error of $|y-C|+ |y-B-C|$. According to Lemma \ref{l12}, $|y-C|+|y-B-C|$ is minimized as long as $C\leq y\leq B+C$, which is the case if all edges are unmarked, i.e. $y=C$. We can use a similar argument if $e^*$ has neither a right nor a left neighbour.
\end{proof}

Observe that marking the left neighbouring edge is not the only way of achieving the lower bound as it can also be achieved by marking the right neighbouring edge. In fact, any assignment of values to $x$ and $y$ such that $x=A+\epsilon_1$, $y=C + \epsilon_2$, $\epsilon_1 + \epsilon_2= B$ will have the same impact. Therefore, for merging a single edge in a weighted path, the marked neighbour can be chosen arbitrarily, and the error value is \textit{oblivious} to the marking direction. However, this observation (being oblivious to the marking direction) only holds for merging two regular nodes. As we will show in Lemma \ref{general}, for merging two supernodes the optimal error is obtained by marking the edge adjacent to the smaller node with respect to cardinality.
\begin{figure*}
\centering
 \subfloat[The original graph, before contracting $e^*= (v_{n_1 +1}, v_{n_1 +2}$) with weight $w^*$]{{\includegraphics[width=305pt]{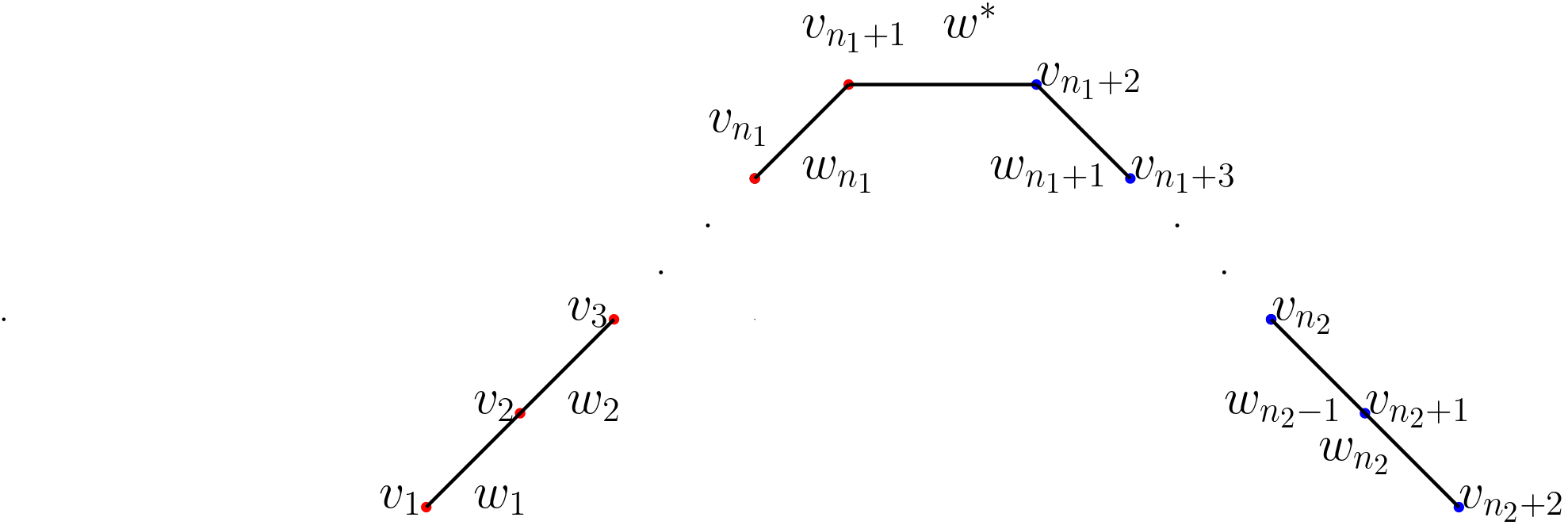}}}
 \qquad
 \subfloat[An arbitrary weight redistribution which is transformed in the proof of Lemma \ref{construction} to the redistribution of Figure \ref{fig12}-(b). ]{{\includegraphics[width=305pt]{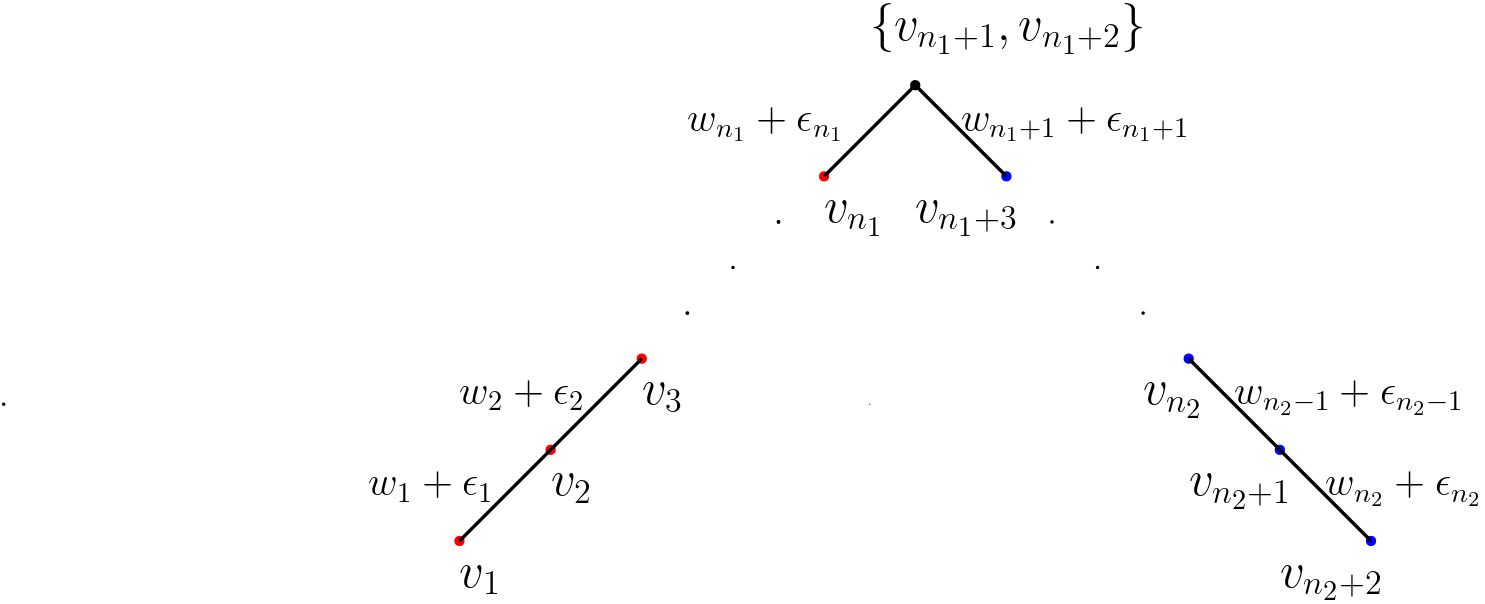}}}
  \caption{The figure used in the proof of Lemma \ref{construction}, the vertices of $V_L$ and $V_R$ are depicted in red and blue respectively. (a) The original graph, before contracting $e^*= (v_{n_1 +1}, v_{n_1 +2}$) with weight $w^*$, (b) An arbitrary weight redistribution which is transformed in the proof of Lemma \ref{construction} to the redistribution of Figure \ref{fig12}-(b). }
\label{fig2}
    \end{figure*} 
    \begin{figure*}[h]
 \subfloat[Before applying the construction method of Lemma \ref{construction}]{{\includegraphics[width=260pt]{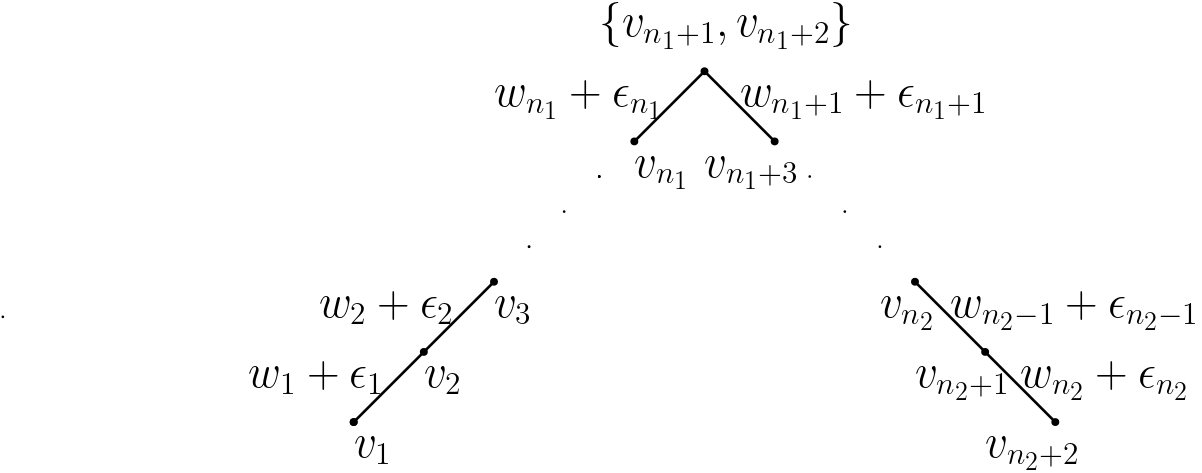}}}
 \subfloat[After applying the construction method of Lemma \ref{construction}]{{\includegraphics[width=260pt]{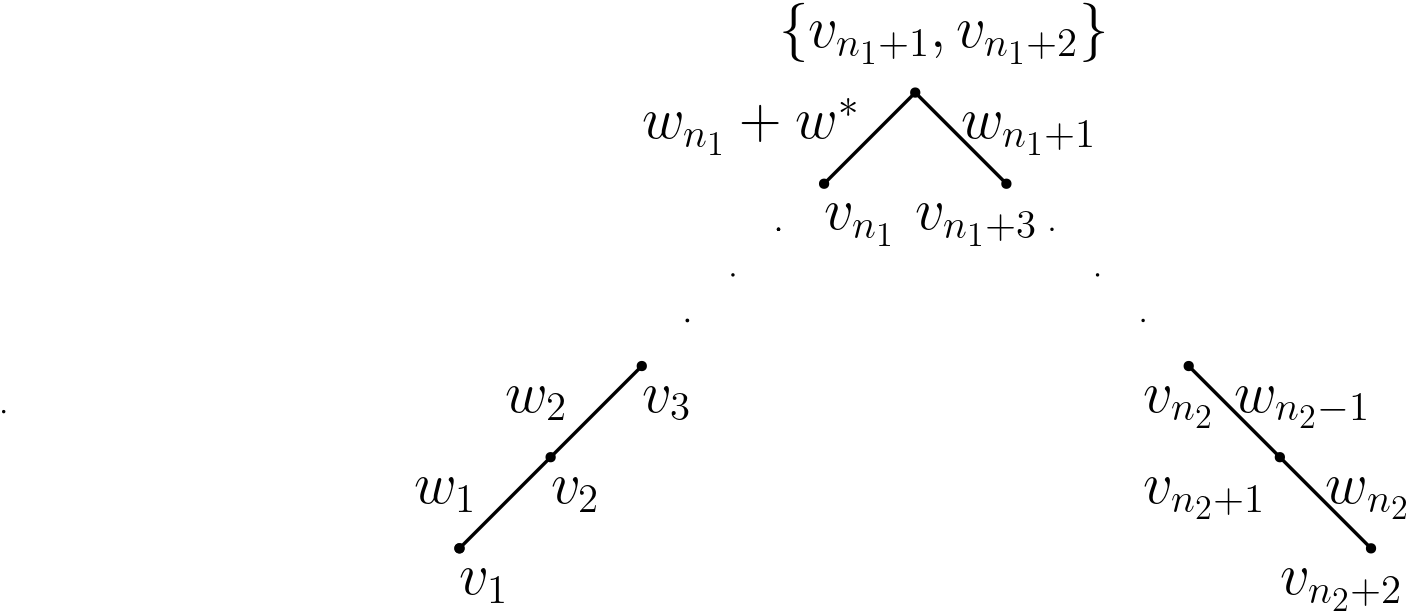}}}
    \caption{The construction method of Lemma \ref{construction}, (a) Before applying the construction and (b) After applying the construction }
    \label{fig12}
\end{figure*} 
 
 \begin{figure*}[h]
\subfloat[The original graph]{{\includegraphics[width=240pt]{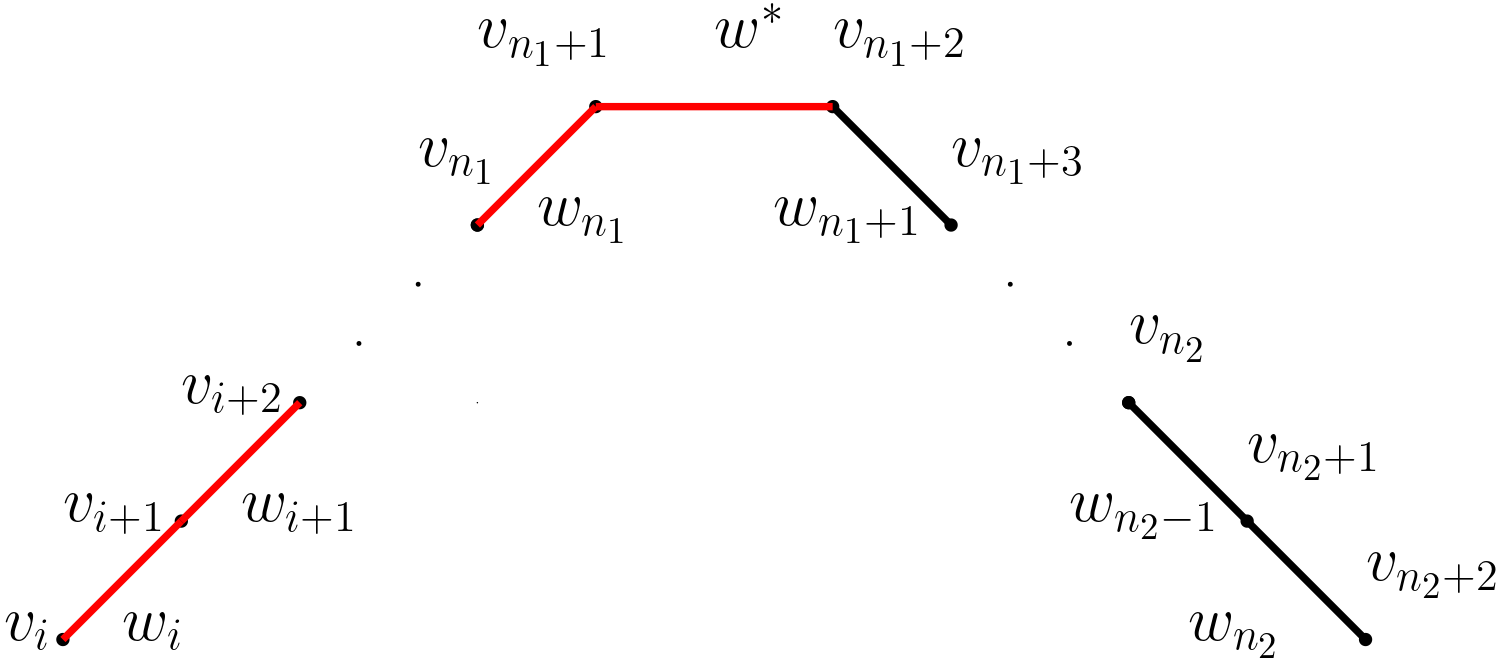}}}
 \subfloat[The original weight redistribution]{{\includegraphics[width=290pt]{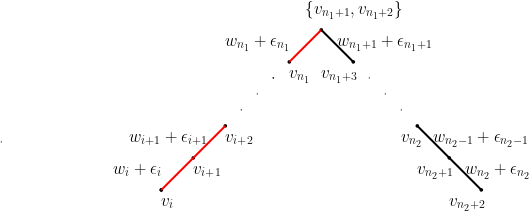}}}

\hspace{6cm}\subfloat[After applying the construction method of Lemma \ref{construction}]{{\includegraphics[width=200pt]{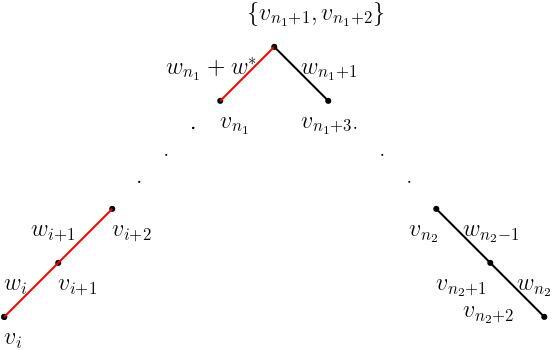}}}
    \caption{Case 2 in the proof of Lemma \ref{construction}: (a) The original graph (b) The original weight redistribution (c) After applying the construction method of Lemma \ref{construction}. The affected shortest path of Case 2 is highlighted in red. }
    \label{fig13}
\end{figure*}
Theorem \ref{l3aval} assumes that to achieve the minimum amount of error, we have to alter only the immediate edges directly connected to the endpoints of the merged edge. We now prove the correctness of this assumption. For modelling the proof, we now define some notation. For any edge $e_i \in E$, we denote its new weight as $w^{\prime}(e_i)=w(e_i) + \epsilon_i$ where $\epsilon_i$ is a real number (see Figure \ref{fig2}). This definition allows us to increase or decrease the weight of any given edge $e_i$ by $\epsilon_i$. We call this assignment of weights a \textit{redistribution} for $e^{*}$. We refer to an edge $e_i$ as \textit{altered} if $\epsilon_i \neq 0$, and \textit{unaltered} otherwise. Moreover, let $V_L, V_R \subset V $ be the vertices to the left and right of the merged edge respectively as depicted in Figure \ref{fig2}. Therefore, the problem is now to show that there exists an optimal solution, with only the immediate edges altered. For simplicity, we slightly abuse the notation and write $w(e_i)$ as $w_i$ and $w(e^{*})$ as $w^{*}$. In the following lemma, we show a construction for transforming any redistribution into another equivalent redistribution in which only the immediate edges are altered.

\begin{lemma}
    (See Figure \ref{fig2} and Figure \ref{fig12}) For a merged edge $e^{*}$ (in a weighted path) that has both left and right neighbouring edges, any weight redistribution can be transformed into another weight redistribution in which only the left neighbouring edge of $e^*$ is altered ($\epsilon_i =0\; \; \forall i \neq n_1 $ in Figure \ref{fig2}). The error associated with this redistribution is no worse than that of the original one. 
    \label{construction}
\end{lemma}

\begin{proof}
    We prove the lemma by presenting a construction method for transforming any arbitrary weight redistribution to another one in which only the left neighbouring edge is altered. Furthermore, we show this transformation does not worsen the error. The illustration is mainly based on Figure \ref{fig2} and Figure \ref{fig12}. Figure \ref{fig2}-(b) depicts an arbitrary weight redistribution for merging $e^*=(v_{n_1 +1}, v_{n_1 +2})$, which is transformed into another weight redistribution (depicted in Figure \ref{fig12}-(b)). 
    
    We now present a simple construction as follows. For 
    illustration, see Figure \ref{fig12}. Set $\epsilon_i 
    =0\; \; \forall i \neq n_1$, and $\epsilon_{n_1}= 
    w^{*}$. Note that this new redistribution may cause 
    some parts of the error to increase. However, we will 
    use Corollary \ref{obs}  to provide an upper bound on 
    any potential error increase and show that there will 
    always be enough decrease in error to counterbalance 
    the increase. In the original redistribution, the error 
    between two vertices $v_i, v_{j} \in V_L (i< j)$ is:
    \begin{linenomath*}\begin{equation}
        \left|\sum_{k=i}^{j-1} (w_k + \epsilon_k)- 
        \sum_{k=i}^{j-1} w_k \right|= \left|\sum_{k=i}^{j-1} 
        \epsilon_k \right|
        \label{sameside}
    \end{equation}\end{linenomath*}
   The indices $i$ and $j$ used in this proof are based on the ones depicted in Figure \ref{fig2} and Figure \ref{fig12}. Assume that the path of Figure \ref{fig2} has $n_2+1$ edges with $V_L=\{v_i| 1 \leq i \leq n_1 +1\}$, $V_R=\{v_i|n_1+2 \leq i\leq n_2 +2\}$, $n_1=n_L$ and $n_2=n_R$. For instance, $v_1$ and $v_3$ are two vertices from $V_L$ in Figure \ref{fig2}. The original shortest path length between $v_1$ to $v_3$ is $w_1+w_2 =\sum_{k=1}^{2}w_k$ (Figure \ref{fig2}-(a)). In the original redistribution of Figure \ref{fig2}-(b) (which we transform into Figure \ref{fig12}-(b)), this length is $w_1+ \epsilon_1 + w_2 + \epsilon_2 =\sum_{k=1}^{2}(w_k +\epsilon_k)$, resulting in Eq. (\ref{sameside}). A similar equation can also be defined for any two vertices $v_i, v_j \in V_R$. Moreover, the error between a vertex $v_i \in V_L$ and a vertex $v_{j} \in V_R$ ($i<j$) in the original redistribution is equal to: 
    \begin{linenomath*}\begin{equation}
         \left|\sum_{k=i}^{j-2} (w_k + \epsilon_k)- w^{*}-\sum_{k=i}^{j-2} w_k \right|= \left| w^{*}-\sum_{k=i}^{j-2} \epsilon_k \right|
         \label{twosides}
    \end{equation}\end{linenomath*}
   Transforming the weight redistribution (as shown in Figure \ref{fig12}) changes the error value. We break this change down into five different cases:
    \begin{itemize}
        \item \textbf{Case 1:} The error between two 
        vertices $v_i, v_{j} \in V_L \;(i<j,\; j \neq n_{1}+1)$ decreases by $-\left |\sum_{k=i}^{j-1} 
        \epsilon_k \right|$, because after the construction 
        this error is equal to zero (compare Figure \ref{fig12}-(b) with Figure \ref{fig2}-(a) for $v_1$ and $v_3$) and using Eq. 
        (\ref{sameside}), the change is equal to $0- 
        \left|\sum_{k=i}^{j-1} \epsilon_k \right|$.  
\item \textbf{Case 2:} The error 
between some vertex $v_i \in V_L, v_i \neq v_{n_1+1}$, and 
every vertex $v_j \in V_R$ decreases. Specifically, the 
error between $v_i$ and $v_{n_1 +2}$ decreases by $-
\left|w^{*}-\sum_{k=i}^{n_1} \epsilon_k \right|$ using Eq. 
(\ref{twosides}) (see Figure \ref{fig13}).

\item \textbf{Case 3:} The error between some vertex $v_i \in V_L, v_i \neq v_{n_1+1}$, and $v_{n_1+1}$ changes by $|w^{*}| - \left| \sum_{k=i}^{n_1} \epsilon_k \right|$, because in the original redistribution, the error is equal to $|\sum_{k=i}^{n_1} \epsilon_k|$ (Eq. (\ref{sameside})) and in the new redistribution, it is equal to $|w^*|$. This change might lead to an increase in error; however, by using Corollary \ref{obs} and setting $z=w^{*}$ and $x=\sum_{k=i}^{n_1} \epsilon_k$ we have:
\begin{linenomath*}$$|w^{*}|-\left| \sum_{k=i}^{n_1} \epsilon_k \right| \leq \left| w^{*}- \sum_{k=i}^{n_1} \epsilon_k \right|$$\end{linenomath*}
In other words, if the construction causes an increase in error, it is at most equal to $\left| w^{*}- \sum_{k=i}^{n_1} \epsilon_k \right|$. However, from Case 2 we know that each such vertex $v_i$ also has an error decrease of $-\left| w^{*}-\sum_{k=i}^{n_1} \epsilon_k \right|$, which will be enough to nullify this increase. 
\item \textbf{Case 4:} The error between some vertex $v_{j} \in V_R, v_{j} \neq v_{n_1 +2}$, and $v_{n_1 +2}$ decreases by $-\left| \sum_{k=n_1+1}^{j-2} \epsilon_k \right|$ using Eq.~(\ref{sameside}). 
\item \textbf{Case 5:} The error between some vertex $v_{j}  \in V_R, v_{j}  \neq v_{n_1 +2}$, and $v_{n_1+1}$ changes by $|w^{*}|- \left| w^{*}-\sum_{k=n_1+1}^{j-2} \epsilon_k \right|$. This change may lead to an increase in error; however, we use Corollary \ref{obs} to bound this increase:
\begin{linenomath*}$$|w^{*}|- \left| w^{*}-\sum_{k=n_1+1}^{j-2} \epsilon_k \right| \leq \left|\sum_{k=n_1+1}^{j-2} \epsilon_k \right|$$\end{linenomath*}
Therefore, we have enough decrease from Case 4 to nullify this increase. 
\end{itemize}
\end{proof} 
\begin{figure*}[h]
    \centering
        \subfloat[]{{\includegraphics[width=250pt]{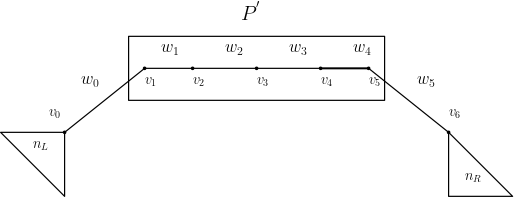}}}
        
        \subfloat[]{{\includegraphics[width=200pt]{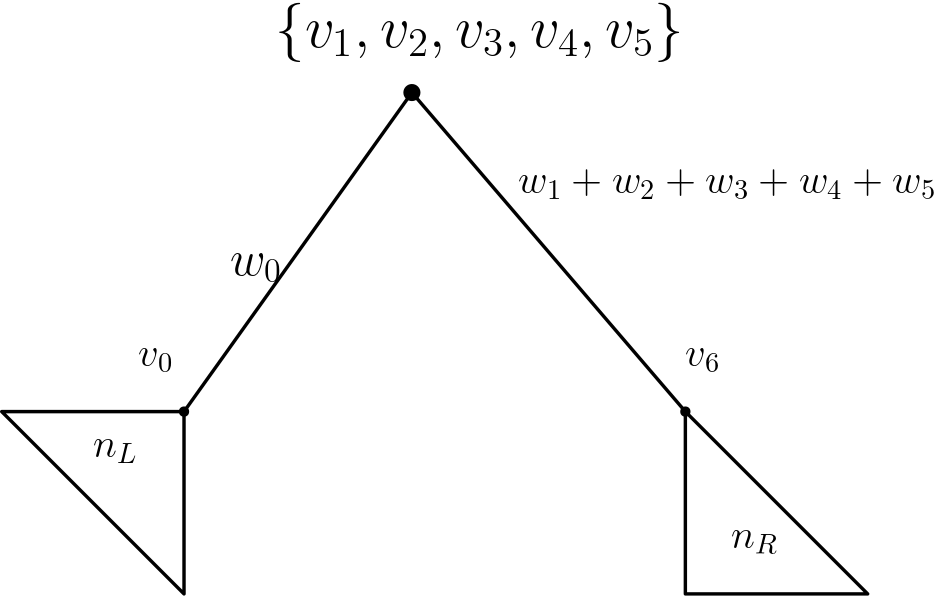}}}
        \subfloat[]{{\includegraphics[width=200pt]{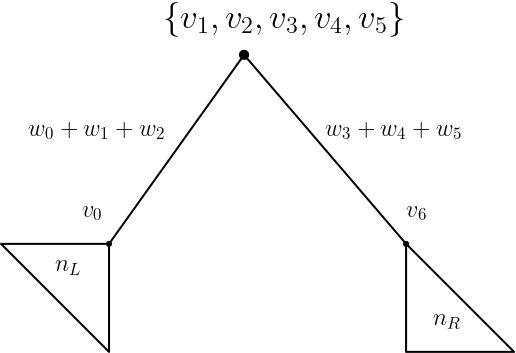}}}
    \caption{(a) An example of merging an entire subpath $P^{\prime} \subset P$ with four edges, (b) A suboptimal solution generated by Algorithm \ref{pathsalg}, and (c) The optimal solution. }
    \label{fig9}
\end{figure*}
\begin{figure*}[h]
    \centering
        \subfloat[Before merging the supernodes]{{\includegraphics[width=220pt]{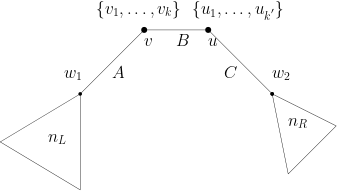}}}
        \subfloat[After merging the supernodes]{{\includegraphics[width=220pt]{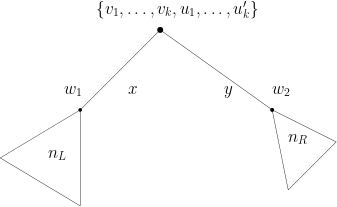}}}\

    \caption{An example of merging two supernodes with cardinalities $k$ and $k^{\prime}$ (a) Before the merge and (b) After the merge}
    \label{fig8}
\end{figure*}
\begin{corollary}
        For a given merged $e^{*}$ with left and right neighbouring edges, there exists an optimal redistribution in which only the left-neighbouring edge of $e^*$ is marked and all other edges are unmarked. 
    \label{maybenot}
\end{corollary}

\begin{algorithm}
\caption{Graph Compression Algorithm for Independent Edges}
\label{pathsalg}
\begin{algorithmic}[1]
\Procedure{PATH\_COMPRESSION}{}    
\State \textbf{Input:} $G=P_n$ (A path of $n$ vertices) with weight function  $w: E \rightarrow \mathbb{R}_{\geq 0}$, a set $E_m$ of edges to merge 
\State $S\xleftarrow[]{} E_m$ 
\While{$S$ is not empty}
\State Pick an arbitrary edge $e \in S$
\If{$e$ has a left neighbouring edge}
\State Let $e^{\prime}$ be the left neighbouring edge of $e$
\State Set $w(e^{\prime}) \xleftarrow[]{} w(e^{\prime}) + w(e)$ \Comment{Mark $e^{\prime}$}
\EndIf
\State Remove $e$ from $G$ and merge its endpoints
\State Remove $e$ from $S$
\EndWhile
\EndProcedure
\end{algorithmic}
\end{algorithm}

Based on Theorem \ref{l3aval}, the algorithm for merging a set $S$ of edges in a path $G=P_n$ is presented in Algorithm \ref{pathsalg}. Algorithm \ref{pathsalg} continuously applies Theorem \ref{l3aval} and marks the left neighbouring edge of each edge $e^* \in S$ taken in an arbitrary order.

Unfortunately, Algorithm \ref{pathsalg} may produce suboptimal results when applied to specific kinds of inputs. Precisely, it may produce suboptimal results when merging a connected subpath of the given path. The reason behind this suboptimal performance lies in the difference between merging two regular nodes and two supernodes. An example of merging a subpath of size four is depicted in Figure \ref{fig9}-(a), for which Algorithm \ref{pathsalg} may produce the suboptimal solution Figure \ref{fig9}-(b). Later in Section \ref{subp}, we shall show that the optimal solution for this example is the one depicted in Figure \ref{fig9}-(c). Furthermore, in Theorem \ref{nobetter}, we prove that when the input to Algorithm \ref{pathsalg} consists of independent edges in $G$ (if $S$ induces a matching in $G$), it produces the optimal results. 

\subsection{Merging Supernodes}
\label{super}
As seen in the previous section, Algorithm \ref{pathsalg} may find suboptimal solutions when given an entire subpath of size $k$. The main reason behind this suboptimal performance lies in the difference between merging regular nodes and \textit{supernodes}. 
 Recall from Definition \ref{superr} that a supernode contains more than one node of the original graph. In this section, we show that merging two supernodes differs from merging two regular vertices, and we provide a generalized version of Theorem \ref{l3aval}. Interestingly, we observe that, unlike merging two regular nodes in which the error value was oblivious to the marking direction, for merging supernodes this direction is directly affected by the cardinality (Definition \ref{superrr}) of each endpoint. In the following lemma, we shall see that for merging an edge $e^*= (u,v)$ connecting two supernodes $u$ and $v$, the optimal solution is obtained by marking the edge adjacent to the lighter vertex (the one with the smaller cardinality) among $u$ and $v$. 
 \newpage
\begin{lemma}    
Suppose we have supernodes v and u (as shown in Figure \ref{fig8}), with $\mathcal{C}(v) =k$ and $\mathcal{C}(u) =k^{\prime}$ (where $k\geq k^{\prime}$), connected to vertices $w_1$ and $w_2$, respectively. The error incurred by merging the edge $e^{*} = (u, v)$ (with weight $B$) is at least $B \times k^{\prime} \times (n-(k+ k^{\prime}))$. Furthermore, this lower bound can be achieved by marking the neighbouring edge adjacent to the smaller vertex among $v$ and $u$ in terms of cardinality ($e = (u, w_2)$ in Figure \ref{fig8}). If the smaller vertex, with reference to cardinality, has no neighbouring edge other than $e^{*} = (u, v)$, then the optimal error can be achieved by contracting $e^{*}$ without any further modifications or weight changes.
    \label{general}
\end{lemma}
\begin{proof}
    The analysis is similar to the case of merging regular vertices, we enumerate all possible error values and then deduce the optimal assignment. We first assume that $u$ (the smaller vertex) is adjacent to another edge $e^{\prime}\neq e^*$, and we handle the other case (only adjacent to $e^*$) later in the proof. We denote the error by $|\Delta E| $. Note that $n_L +n_R= |\overline{V_m}|=n-(k+k^{\prime})$. Let $x$ and $y$ denote the new weights of the edges adjacent to $(v,u)$, we have:
    \begin{linenomath*}
\begin{align*}
       & |\Delta E| =\underbrace{n_L \times |x-A| \times k}_{\text{between the subpath of } w_1 \text{ and the vertices in } v}+ \underbrace{n_L \times|x-A-B|\times k^{\prime}}_{\text{between the subpath of } w_1 \text{ and the vertices in } u}+  \\
    &\underbrace{n_R \times |y-C| \times k^{\prime}}_{\text{between the subpath of } w_2 \text{ and the vertices in } u} + \underbrace{n_R \times|y-B-C|\times k}_{{\text{between the subpath of } w_2 \text{ and the vertices in } v}} + \\
    & \underbrace{n_L \times n_R \times|x+y -A-B-C|}_{\text{between the subpath of } w_1\text{ and }w_2}
\end{align*}
\end{linenomath*}
Because $k\geq k^{\prime}$, we further simplify $|\Delta E|$ as:
\begin{linenomath*}
\begin{align}
    |\Delta E| =& (k-k^{\prime}) \times n_L \times |x-A| + n_L \times k^{\prime}\times \big(|x-A| + |x-A-B|\big)   \nonumber\\
     +&(k-k^{\prime}) \times n_R \times |y-B-C| + n_R \times k^{\prime}\times \big(|y-C| + |y-B-C|\big) \nonumber \\
     +&n_L \times n_R \times|x+y -A-B-C|
     \label{khaste}
\end{align}\end{linenomath*}
Using Lemma \ref{l12}, and the fact that $k-k^{\prime} \geq 0$, we have: 
\begin{linenomath*}\begin{align}
        |\Delta E| \geq &  n_L \times k^{\prime}\times \underbrace{\big(|x-A| + |x-A-B|\big)}_{\geq B \text{ (Lemma \ref{l12})}} +  n_R \times k^{\prime} \times \underbrace{\big(|y-C| + |y-B-C|\big)}_{\geq B \text{ (Lemma \ref{l12})}}  \label{hu}\\
     \geq&  B \times k^{\prime}\times \underbrace{(n_L + n_R )}_{=n-(k+ k^{\prime})} =  B \times k^{\prime} \times (n-(k+ k^{\prime}))\label{hu2}
\end{align}\end{linenomath*}
 We can observe that this lower bound (Eq. (\ref{hu2})) is tight and can be achieved by setting $y=B+C$ and $x=A$ in Eq. (\ref{khaste}).

Now if $u$ is only adjacent to $e^*$, the lower bound can be achieved by just contracting $e^*$ and leaving the weight function unchanged. To see why, suppose $u$ is only adjacent to $e^*$. We have $n_R=0$ and the error is equal to: 
\begin{linenomath*}\begin{align*}
    |\Delta E| =&  (k-k^{\prime}) \times n_L \times |x-A| + n_L \times k^{\prime}\times \big(|x-A| + |x-A-B|\big)   \\
     +&(k-k^{\prime}) \times n_R \times |y-B-C| + n_R \times k^{\prime}\times \big(|y-C| + |y-B-C|\big)\\
     +&n_L \times n_R \times|x+y -A-B-C|\xrightarrow{n_R=0}\\
     =&(k-k^{\prime}) \times n_L \times |x-A| + n_L \times k^{\prime}\times \big(|x-A| + |x-A-B|\big)\xrightarrow{\text{no edge weights are changed, set }x=A}\\
     =&n_L \times k^{\prime} \times B
\end{align*}\end{linenomath*}
On the other hand, note that $n_L+n_R=n-(k+ k^{\prime})$, and $n_R=0$ implies that $n_L=n-(k+ k^{\prime})$. We have:
\begin{linenomath*}$$|\Delta E|=n_L \times k^{\prime} \times B=(n-(k+ k^{\prime})) \times k^{\prime} \times B$$\end{linenomath*}
and the lower bound is achieved without any weight changes.

It is worth noting that similar to Theorem \ref{hehe}, we assume that it is sufficient to only alter the neighbouring edges of the merged edge $e^*=(u,v)$. This proof for this assumption is almost identical to that of Lemma \ref{construction}, where any arbitrary redistribution can be transformed into another redistribution in which only the edge adjacent to the smaller vertex ($e=(u,w_2)$) is marked. Then, similar to the proof of Lemma \ref{construction}, the decrease in error is always sufficient to counterbalance any potential error increase. The only difference is that in the new proof, the decrease in error and any potential error increase are weighted by $k$ and $k^{\prime}$, respectively. Since $k \geq k^{\prime}$, the proof follows. 
\end{proof}

\begin{remark}
 Lemma \ref{general} is a generalization of Theorem \ref{l3aval}. Thinking of each regular vertex as a supernode with cardinality one, we have $k=k^{\prime}=1$ and using Lemma \ref{general}, the error is equal to $ B \times k^{\prime} \times (n-(k+ k^{\prime}))= B \times (n-2)$ by arbitrarily marking one of the neighbouring edges (since the endpoints have equal cardinalities). 
 \end{remark}
 \begin{remark}
 Using Lemma \ref{general}, we can now explain the suboptimal performance of Algorithm \ref{pathsalg} for edges that are not independent and form a contiguous subpath. For inputs of such kind, Algorithm \ref{pathsalg} continuously marks the left neighbouring edges of all edges $e^* \in S$, potentially marking an edge adjacent to the heavier endpoint of some $e^* \in S$ along the way and violating the conditions of Lemma \ref{general}.
\end{remark}
 In the next section, we study the problem of optimally merging an entire contiguous subpath of the path. 
\subsection{Merging Contiguous Subpaths}
\label{subp}
\begin{figure*}[h]
    \centering
        \subfloat[]{{\includegraphics[width=280pt]{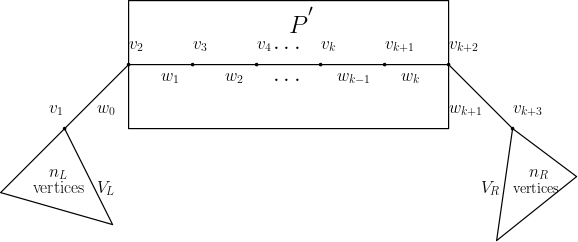}}} \qquad
\subfloat[]{{\includegraphics[width=220pt]{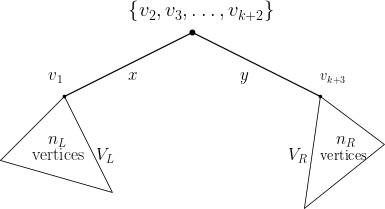}}}
    \caption{The figure used in the proofs of Section \ref{subp} (a) Before merging an entire subpath $P^{\prime} \subset P$ with $k$ edges, (b) After the merge. }
    \label{fig10}
\end{figure*}
This section presents an optimal way of merging any contiguous subpath (or connected subpath) of a given path. For convenience, we refer to contiguous subpaths as subpaths. Let $P ^{'} \subseteq P$ be the desired subpath consisting of $k$ edges (see Figure \ref{fig10} for an illustration). Throughout this section, we assume $k$ is even; otherwise, we can convert $P^{\prime}$ into an equivalent subpath of even length by adding a dummy edge of weight zero. As depicted in Figure \ref{fig10}, we assume $P^{\prime}$ partitions the set of vertices into two subsets, $V_L$ and $V_R$, with $n_L$ and $n_R$ vertices respectively. We denote the error associated with contracting $P^{\prime}$ by $\mathcal{E}$ and break it down into three components:
\begin{itemize}
    \item $\mathcal{E}_L$, the error between the vertices in $V_L$ and the ones inside $P^{\prime}$,
    \item $\mathcal{E}_R$, the error between the vertices in $V_R$ and the ones inside $P^{\prime}$, and
        \item $\mathcal{E}_{LR}$ the error between the vertices of $V_L$ and $V_R$.
\end{itemize}
With this in mind, we formulate $\mathcal{E}$ as: 
\begin{linenomath*}\begin{equation}
    \mathcal{E}=\mathcal{E}_L +\mathcal{E}_R +\mathcal{E}_{LR}
    \label{subpath}
\end{equation}\end{linenomath*}
such that:
\begin{linenomath*}\begin{equation}
    \mathcal{E}_L=n_L \times \big( \underbrace{|x-w_0|}_{\text{between the vertices of } V_L \text{ and } v_2}+ \underbrace{|x-w_0-w_1|}_{\text{between the vertices of } V_L \text{ and } v_3}+\dots + \underbrace{|x-w_0-w_1-\dots -w_k|}_{\text{between the vertices of } V_L \text{ and } v_{k+2}}\big)
    \label{subpath2}
\end{equation}\end{linenomath*}
\begin{linenomath*}\begin{equation}
    \mathcal{E}_R=n_R \times \big( \underbrace{|y-w_{k+1}|}_{\text{between the vertices of } V_R \text{ and } v_{k+2}}+ \underbrace{|y-w_{k+1}-w_{k}|}_{\text{between the vertices of } V_R \text{ and } v_{k+1}}+\dots + \underbrace{|y-w_{k+1}-w_{k}-\dots- w_1|}_{\text{between the vertices of } V_R \text{ and } v_{2}}\big)
    \label{subpath3}
\end{equation}\end{linenomath*}
and 
\begin{equation}
    \mathcal{E}_{LR}=n_L \times n_R \times |x+y - w_0-w_1 - \dots - w_{k+1}|
    \label{subpath4}
\end{equation}
where $x$ and $y$ are the new weights of the neighbouring edges of $P^{\prime}$ (Figure \ref{fig10}-(b)).

We first prove the optimal solution for $\mathcal{E}_L$ and derive the optimal solution for $\mathcal{E}_R$ by symmetry. Let $\mathcal{E}_L ^{(i)}$ denote the value of $\mathcal{E}_L$ when $x=w_0+w_1 +\dots + w_i$ for $0 \leq i \leq k$. We prove the following lemma using induction on $i$.
\begin{lemma}
    $\mathcal{E}_L ^{(i)}= n_L \times \big(\sum_{j=0}^{i} j \; w_j + \sum_{j=i+1}^{k}(k+1-j)\; w_j\big)$
    \label{induction1}
\end{lemma}
\begin{proof}
     For the base case, $\mathcal{E}_L ^{(0)}$, assume $x=w_0$. By a simple replacement into Eq.~(\ref{subpath2}) we get:
\begin{linenomath*}\begin{align*}
   \mathcal{E}_L ^{(0)}= n_L\times \big(w_1 +w_1 +w_2 +w_1+w_2+w_3 +\dots + w_1 +w_2 +\dots +w_k \big)
   \end{align*} \end{linenomath*}
    In other words, every $w_j$, $1 \leq j \leq k$, is repeated $k+1-j$ times, and:
       \begin{linenomath*} $$\mathcal{E}_L ^{(0)}= n_L \times \big( \sum_{j=1}^{k} (k+1 -j)\; w_j  \big)$$\end{linenomath*}
Now assume the lemma holds for all $j<i+1$. By the inductive hypothesis, we have  $\mathcal{E}_L ^{(i)}= n_L \times \big(\sum_{j=0}^{i} j \; w_j + \sum_{j=i+1}^{k}(k+1-j)\; w_j\big)$. We break Eq. (\ref{subpath2}) into $k+1$ \textit{clauses}, such that $c_j=|x-w_0-w_1-\dots - w_j|$ for $0 \leq j \leq k$. Going from $x=\sum_{j=0}^{i}w_j$ to $x=\sum_{j=0}^{i+1}w_j$, $\mathcal{E}_L ^{(i)}$ first increases by $n_L \times \big((i+1)\; w_{i+1}\big) $ because there are $i+1$ clauses $c_0, c_1, \dots, c_i $ that do not include $w_{i+1}$, and then decreases by $ n_L \times \big((k-i) \; w_{i+1} \big)$ because there are $k-i$ clauses $c_{i+1}, \dots, c_{k}$ that include $w_{i+1}$ and were not covered by the previous assignment of $x$ ($x=\sum_{j=0}^{i}w_j$). Therefore, we have: 
\begin{linenomath*}\begin{align*}
\mathcal{E}_L ^{(i+1)}=& \mathcal{E}_L ^{(i)} + n_L \times \big((i+1)\; w_{i+1} - (k-i) \; w_{i+1}\big) \\
=& 
n_L \times \big(\sum_{j=0}^{i} j \; w_j + \sum_{j=i+1}^{k}(k+1-j)\; w_j +(i+1)\; w_{i+1} - (k-i) \; w_{i+1}\big)\\
=& n_L \times \big(\sum_{j=0}^{i+1} j \; w_j + \sum_{j=i+2}^{k}(k+1-j)\; w_j\big)
\end{align*}\end{linenomath*}
\end{proof}
The following lemma states that the optimal value of $\mathcal{E}_L$ is equal to $\mathcal{E}_L^{(\frac{k}{2})}$.
\begin{lemma}
The optimal value of $\mathcal{E}_L$ is obtained when $x=w_0+w_1+\dots+w_{\frac{k}{2}}$.
    \label{induction2}
\end{lemma}
\begin{proof}
    It suffices to show the optimal value of $\mathcal{E}_L$ is equal to $\mathcal{E}_L ^{(\frac{k}{2})}$. From the proof of Lemma \ref{induction1}, we know that $\mathcal{E}_{L}^{(i+1)}-\mathcal{E}_{L}^{(i)}=n_L \times \big ((i+1) \; w_{i+1} - (k-i) \; w_{i+1} \big)$. Therefore, $\mathcal{E}_{L}^{(i+1)}-\mathcal{E}_{L}^{(i)} < 0 $ if:
    \begin{linenomath*}$$i+1 -k+i <  0 \xrightarrow[]{} 2i < k-1 \xrightarrow[]{} i < \frac{k}{2} - \frac{1}{2} \xrightarrow[\text{since } k \text{ is even}]{} i \leq \frac{k}{2}-1$$\end{linenomath*}
    In other words, $\mathcal{E}_{L}^{(\frac{k}{2})}$ is strictly better than (less than) any $\mathcal{E}_{L}^{(j)}, j\neq \frac{k}{2}$. Note that the optimal solution also cannot happen when $x= \epsilon + \sum_{j=0}^{\frac{k}{2}} w_j  $ for some $0< \epsilon < w_{\frac{k}{2}+1}$, because in that case, the error would be equal to: 
    \begin{linenomath*}$$\mathcal{E}_{L}^{(\frac{k}{2})} + (\frac{k}{2} +1) \; \epsilon - (\frac{k}{2}) \; \epsilon > \mathcal{E}_{L}^{(\frac{k}{2})}$$\end{linenomath*} Using simple replacements, we can deduce that $\mathcal{E}_L ^{(\frac{k}{2})}$ is also smaller than $\mathcal{E}_L$ when $x<w_0$ or $x> w_0 +\dots +w_k$. Let $\mathcal{E}^{(x<w_0)}_L$ denote the value of $\mathcal{E}_L$ for some $x<w_0$. For some $x<w_0$, all clauses in Eq. (\ref{subpath2}) have negative values. Recalling that $|x|=-x$ when $x<0$, we have:
    \begin{linenomath*}\begin{align*}
    \mathcal{E}^{(x<w_0)}_L&=n_L\times(w_0-x + w_0 +w_1 -x + \dots +w_0 +w_1 +\dots + w_k -x )= n_L \times \big(\big(\sum_{j=0}^{k}(k+1-j) w_j\big) - (k+1) \times x\big)\\
    \xrightarrow{0 \leq x<w_0}&>n_L \times \big( \sum_{j=1}^{k}(k+1-j) w_j \big)\xrightarrow{\text{See the proof of Lemma \ref{induction1}}}=\mathcal{E}^{(0)}_L >\mathcal{E}^{(\frac{k}{2})}_L
    \end{align*}\end{linenomath*}
    The other case ($x>w_0 +\dots + w_k$) can be handled analogously.
\end{proof}

\begin{lemma}
The optimal value of $\mathcal{E}_R$ is obtained when $y=w_{\frac{k}{2}+1}+w_{\frac{k}{2}+2}+\dots+w_{k+1}$.
    \label{induction3}
\end{lemma}
\begin{proof}
By symmetry and using Lemma \ref{induction1} and Lemma \ref{induction2}.
\end{proof}
We now derive the following theorem, which states that the optimal way of contracting an entire subpath is by distributing the left and right halves of the edges in the subpath to the left and right neighbours respectively. 
\begin{theorem}
        Let $P^{\prime} \subseteq P$ be a contiguous subpath of $P$ (a weighted path on $n$ vertices) consisting of $k$ edges $\{e_1, \dots, e_k\}$, and let $e_0$ and $e_{k+1}$ be the left and right neighbouring edges of $P^{\prime}$ respectively. Furthermore, let $w_i=w(e_i) \; \forall i \in \{ 0,\dots , k+1\}$. The optimal error for contracting $P^{\prime}$ is obtained by setting $x=w_0+w_1+\dots+w_{\frac{k}{2}}$ and $y=w_{\frac{k}{2}+1}+w_{\frac{k}{2}+2}+\dots+w_{k+1}$, where $x$ and $y$ are the new edge weights of $e_0$ and $e_{k+1}$ respectively (see Figure \ref{fig10}). If $P^{\prime}$ has no left neighbour ($e_0$ does not exist), the optimal error can be achieved by setting $y=w_{\frac{k}{2}+1}+w_{\frac{k}{2}+2}+\dots+w_{k+1}$. If $P^{\prime}$ has no right neighbour ($e_{k+1}$ does not exist), the optimal error can be achieved by setting $x=w_0+w_1+\dots+w_{\frac{k}{2}}$. Finally, if $P^{\prime}$ has neither a left nor a right neighbour, the optimal error can be achieved by simply contracting $P^{\prime}$ and no further modifications (weight changes) are required.  
        \label{subp2}
\end{theorem}
\begin{proof}
    The case with both neighbours existing is immediate from Lemma 
    \ref{induction2}, Lemma \ref{induction3}, Eq. (\ref{subpath}), and the fact 
    that $\mathcal{E}_{LR}=0$ when $x=w_0+w_1+\dots+w_{\frac{k}{2}}$ and 
    $y=w_{\frac{k}{2}+1}+w_{\frac{k}{2}+2}+\dots+w_{k+1}$.\\
    If $P^{\prime}$ has no left neighbour ($e_0$ does not exist), we have $n_L=0$ 
    and consequently $\mathcal{E}_{LR}=E_{L}=0$. It follows that 
    $\mathcal{E}=\mathcal{E}_R$ whose optimal value is obtained by setting 
    $y=w_{\frac{k}{2}+1}+w_{\frac{k}{2}+2}+\dots+w_{k+1}$ using Lemma 
    \ref{induction3}. The other cases can be shown analogously. 

    To prove that it is sufficient to alter only the immediate neighbouring edges 
    of $P^{\prime}$, we only provide a sketch to avoid repetition. The idea is 
    very similar to the proof of Lemma \ref{construction} and Lemma \ref{superr}. 
    Suppose we have any arbitrary weight redistribution, which we transform to the 
    one provided in this theorem. Let $u$ be some vertex in $V_L$ (as in Figure 
    \ref{fig10}). In the original redistribution, let $x$ be the length of the 
    shortest path from $u$ to the super vertex $v^{*}=\{v_2, v_3, \dots, 
    v_{k+2}\}$ in $P^{\prime}$ (Figure \ref{fig10}-(b)). It is easy to see that in 
    the original distribution, the error between $u$ and all of the vertices in 
    $v^{*}$ is equal to:
    \begin{linenomath*}$$\mathcal{E}_1=|x-w_0|+\dots+ |w-w_0-\dots- w_k|$$\end{linenomath*}
    For a fixed $x$, $\mathcal{E}_1$ corresponds to $\frac{\mathcal{E}_L}{n_L}$ 
    (Eq. (\ref{subpath2})). It is easy to see that in the new redistribution, the 
    error between $u$ and all vertices in $v^*$ is equal to 
    $\frac{\mathcal{E}^{(\frac{k}{2})}_L}{n_L}$. Therefore, using Lemma 
    \ref{induction2}, we know that $\frac{\mathcal{E}^{(\frac{k}{2})}_L}{n_L} - 
    \frac{\mathcal{E}_L}{n_L} \leq 0$ for any $x$, and this change in the weight 
    redistribution cannot worsen the error associated with any $u \in V_L$. Other 
    cases can be handled analogously.  
\end{proof}

\subsection{Merging A Set of Independent Edges}
\label{indp}
We now generalize the results of Section \ref{oneedgepath} by proving the correctness of Algorithm \ref{pathsalg} for merging any set of independent edges. The proof of correctness consists of the following lemma and theorem which are similar to Lemma \ref{construction} and Theorem \ref{l3aval} respectively.
\begin{figure*}[h]
    \centering
        \subfloat[]{{\includegraphics[width=300pt]{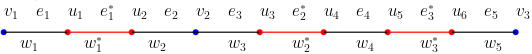}}}
        \hspace{2 cm}
        \subfloat[]{{\includegraphics[width=200pt]{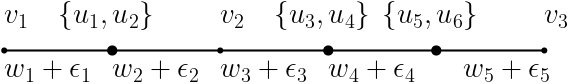}}}
        \hspace{2 cm}
        \subfloat[]{{\includegraphics[width=200pt]{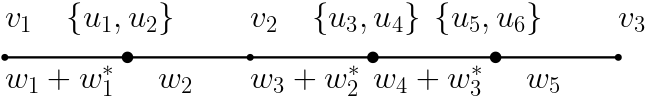}}}
        \hspace{2 cm}
    \caption{The figure used in the proof of Lemma \ref{construction23}. (a) The original graph. The vertices and edges in $V_m$ and $E_m$ are depicted in red, and the vertices in $\overline{V_m}$ are depicted in blue. (b) An arbitrary weight redistribution which assigns $w^{\prime}(e_i) = w(e_i) +\epsilon_i$ to every edge $e_i \in \overline{E_m}=E- E_m$ (c) Another weight redistribution that only marks the left neighbouring edge of each edge in $E_m$ whose associated error is no worse than the one depicted in (b).}
    \label{fig15}
\end{figure*}
\begin{lemma}
    For merging a set of independent edges $E_m$ from a path on $n$ vertices $P_n$, there exists an optimal redistribution in which for each $e \in E_m$, only its left neighbouring edge is marked. If $e^{\prime} \in E_m$ is the leftmost edge on $P_n$, then this optimal solution is obtained by marking the left neighbouring edge of all edges in $E_m$ except for $e^{\prime}$.
    \label{construction23}
\end{lemma}
\begin{proof}
    The proof is similar to the proof of Lemma \ref{construction} and we will provide a sketch using Figure \ref{fig15}. In Figure \ref{fig15}, the edges in $E_m$ and the vertices in $V_m$ are highlighted in red, and the vertices in $\overline{V_m}$ are depicted in blue. We assign an ordering to the vertices (of $V_m$ and $\overline{V_m}$) and the edges (of $E_m$ and $\overline{E_m}$) from left to right, as illustrated in Figure \ref{fig15}. Let $v_i$ and $u_j$ be the $i$-th and the $j$-th vertex in $\overline{V_m}$ and $V_m$ respectively according to this ordering. Similarly, let $e_i$ and $e^{*}_j$ be the $i$-th and the $j$-th edge in $\overline{E_m}=E- E_m$ and $E_m$ respectively. For convenience, we denote $w(e_i)$ and $w(e^{*}_i)$ by $w_i$ and $w^{*}_i$ respectively. Figure \ref{fig15}-(b) depicts some arbitrary weight redistribution in which the new weight of each edge $e_i$ is set to $w(e_i) + \epsilon_i$. We shall show that the error associated with the weight redistribution of Figure \ref{fig15}-(c) (in which the left neighbours of $E_m$ are marked) is no worse than that of Figure \ref{fig15}-(b). We again assume that all edges in $E_m$ have left neighbours. First, observe how this new weight redistribution removes any error between the vertices in $\overline{V_m}$. For instance, in the path of Figure \ref{fig15}-(c), the shortest path value between $v_1, v_3 \in \overline{V_m}$ is the same as the one in the original path (Figure \ref{fig15}-(a)). Therefore, it suffices to study only the error between all pairs of vertices $(u,v), \; \; u \in V_m,\; v \in \overline{V_m} $. Using our ordering of edges, let $e^{*}_k=(u_{j},u_{j+1} ) \in E_m$ and let $v_i \in \overline{V_m}$ be a vertex to the left of $e^*_k$ (we will explain how the other case can be handled analogously). Continuing with our example of Figure \ref{fig15}, let $e^{*}_k=e^{*}_3=(u_5, u_6)$ and $v_i=v_1$. Observe how between $v_i$ ($v_1$ in Figure \ref{fig15}-(a)) and $u_{j+1}$ ($u_6$ in Figure \ref{fig15}-(c)), there exists no error in the new redistribution as they have equal shortest path values in the original graph (Figure \ref{fig15}-(a)) and the new distribution (Figure \ref{fig15}-(c)). We show that, going from the distribution of Figure \ref{fig15}-(b) to the one in Figure \ref{fig15}-(c), any increase in the error between $v_i$ and the left endpoint of $e^*_k$ ($u_{j}$) can be nullified by the decrease in the error between $v_i$ and $u_{j+1}$. The case where $v_i$ is located on the right of $e^*_k$ can be handled similarly. \\ For any $E^{\prime} \subseteq E$ we define the following quantities:
    
    \begin{linenomath*}\begin{equation*}
        \nw(E^{\prime})=\sum_{e \in E^{\prime} \cap \overline{E_m}}w(e), \; \; \nw^*(E^{\prime})=\sum_{e \in E^{\prime} \cap {E_m}}w(e), \; \; \nw^{\prime}(E^{\prime})=\sum_{e_i \in \overline{E_m} \cap E^{\prime}} \epsilon_i
    \end{equation*}\end{linenomath*}
    where $\nw^{\prime}(E^{\prime})$ denotes the sum of all $\epsilon_i$'s in the distribution of Figure \ref{fig15}-(b). Let $\pi_{v, u}$, $\pi^{\prime}_{v, u}$, and $\pi^{\prime \prime}_{v, u}$ denote the shortest path values between $v$ and $u$ in the original graph (Figure \ref{fig15}-(a)), the first redistribution (Figure \ref{fig15}-(b)), and the second redistribution (Figure \ref{fig15}-(c)) respectively. Moreover, let $E^{(u,v)}$ denote the set of edges on the unique shortest path from $u$ to $v$. We have: 
    
  \begin{linenomath*}  \begin{align}
    \pi_{v_i, u_j}&= \nw(\eij)+ \nw^*(\eij) \label{oneij}\\ 
    \pi^{\prime}_{v_i, u_j}&=\nw(\eij)+\nw^{\prime}(\eij) \label {twoij}\\
    \pi^{\prime \prime}_{v_i, u_j}&=\nw(\eij)+\nw^{*}(\eij)+ w^*_k \label {threeij}
    \end{align}\end{linenomath*}
    \\\\\\
We provide some examples of these quantities in Example \ref{ex} for better readability. 

Note that:
\begin{linenomath*}\begin{equation}
        \pi_{v_i, u_{j+1}}= \pi_{v_i, u_{j}}+w^*_k \text{,  } \pi^{\prime}_{v_i, u_j}=\pi^{\prime}_{v_i, u_{j+1}}\text{, and } \pi^{\prime \prime}_{v_i, u_j}=\pi^{\prime \prime}_{v_i, u_{j+1}}
    \label{fourij}
\end{equation}\end{linenomath*}
The error between $v_i$ and $u_{j+1}$ in the redistribution of Figure \ref{fig15}-(b) is: 
\begin{linenomath*}\begin{equation}
   \mathcal{E}^{v_i, u_{j+1}}_1=\left|\pi_{v_i,u_{j+1}}-\pi^{\prime}_{v_i, u_{j+1}}\right|=\left|  \pi_{v_i, u_{j}}+w^*_k -\pi^{\prime}_{v_i, u_j}\right| = \left|w^{*}_k+\nw^{*}(\eij)- \nw^{\prime}(\eij)\right|
\end{equation}\end{linenomath*}
As mentioned before, the error between $v_i$ and $u_{j+1}$ in the weight redistribution of Figure \ref{fig15}-(c) is equal to zero: 
\begin{equation}
     \mathcal{E}^{v_i, u_{j+1}}_2=0
\end{equation}
    Therefore, transforming Figure \ref{fig15}-(b) into Figure \ref{fig15}-(c) changes the error between $v_i$ to $u_{j+1}$ by: 
    \begin{linenomath*}\begin{equation}
         \Delta_{v_i, u_{j+1}}=\mathcal{E}^{v_i, u_{j+1}}_2-  \mathcal{E}^{v_i, u_{j+1}}_1=-\left|w^{*}_k+\nw^{*}(\eij)- \nw^{\prime}(\eij)\right|
    \end{equation}\end{linenomath*}
    
The error between $v_i$ and $u_{j}$ in the redistribution of Figure \ref{fig15}-(b) is: 
\begin{linenomath*}\begin{equation}
   \mathcal{E}^{v_i, u_{j}}_1=\left|\pi_{v_i,u_{j}}-\pi^{\prime}_{v_i, u_{j}}\right|=\left|\pi^{\prime}_{v_i,u_{j}}-\pi_{v_i, u_{j}}\right|=  \left|\nw^{\prime}(\eij)- \nw^{*}(\eij)\right|
\end{equation}\end{linenomath*}
The error between $v_i$ and $u_{j}$ in the weight redistribution of Figure \ref{fig15}-(c) is equal to:
\begin{linenomath*}\begin{equation}
     \mathcal{E}^{v_i, u_{j}}_2=\left|\pi^{\prime \prime}_{v_i, u_{j}}-\pi_{v_i, u_{j}}\right|=\left|w^*_k\right|
\end{equation}\end{linenomath*}
Transforming Figure \ref{fig15}-(b) into Figure \ref{fig15}-(c) changes the error between $v_i$ to $u_{j}$ by: 
    \begin{linenomath*}\begin{equation}
         \Delta_{v_i, u_{j}}=\mathcal{E}^{v_i, u_{j}}_2-  \mathcal{E}^{v_i, u_{j}}_1=\left|w^*_k\right|-\left|\nw^{\prime}(\eij)- \nw^{*}(\eij)\right|\leq \left|w^*_k-\nw^{\prime}(\eij)+ \nw^{*}(\eij)\right|
    \end{equation}\end{linenomath*}
using Corollary \ref{obs}. Therefore, going from the first redistribution to the second one changes the error between the endpoints of $e^*_k=(u_{j}, u_{j+1})$ and $v_i$ by:
\begin{linenomath*}$$\Delta_{v_i, u_{j}}+\Delta_{v_i, u_{j+1}}\leq \left|w^*_k-\nw^{\prime}(\eij)+ \nw^{*}(\eij)\right|-\left|w^{*}_k+\nw^{*}(\eij)- \nw^{\prime}(\eij)\right| \leq 0$$\end{linenomath*}
Since each $e^*_k$ edge in $E_m$ has exactly two endpoints, this concludes the proof for the first case ($v_i$ is on the left of $e^*_k$). The other case can be handled analogously.
\end{proof}
\begin{example}
\label{ex}
Returning to our example of Lemma \ref{construction23} and Figure \ref{fig15}, let $e^*_k=(u_{j}, u_{j+1})= e^*_3=(u_5, u_6)$, and $v_i=v_1$. Then:
\begin{itemize}
    \item $E^{(v_1, u_5)}=\{e_1, e^*_1, e_2, e_3, e^*_2, e_4\}$
    \item $\nw (E^{(v_1, u_5)})= w_1+w_2+w_3+w_4$
    \item $\nw^* (E^{(v_1, u_5)})= w^*_1+w^*_2$
    \item $\nw^{\prime} (E^{(v_1, u_5)})=\epsilon_1+\epsilon_2+\epsilon_2+\epsilon_3+\epsilon_4$

\end{itemize}
\end{example}
\begin{theorem}
   Let $|\Delta E|$ be the optimal error resulting from merging a set of $k$ independent edges $e_1, e_2,\dots, e_k$ with respective weights $w^*_1, w^*_2, \dots, w^*_k$ from a path on $n$ vertices $P_n$. Let $(u_{2i-1}, u_{2i})$ be the endpoints of $e_i \in E_m, 1\leq i \leq k$. Furthermore, let $V_m=\{u_1,\dots, u_{2k}\}$ and $\overline{V_m}=V-V_m$. We have $|\Delta E|= |\overline{V_m}| (w^*_1+\dots+w^*_k)= (n-2k)(w^*_1+\dots+w^*_k)$. This optimal value can be achieved by marking the left neighbour of each edge in $E_m$ after contraction. If the leftmost edge in $E_m$ has no left neighbour, the optimal error can be achieved by marking the left neighbours of all other edges in $E_m$.
\label{nobetter}
\end{theorem}
\begin{proof}
Let $ w^{\prime}: E \rightarrow \mathbb{R}_{\geq 0}$ be the weight redistribution that marks the left neighbouring edge (if any) of each edge in $E_m$ (Figure \ref{fig15}-(c)). That $w^{\prime}$ is optimal follows directly from Lemma \ref{construction23}. We now prove the error associated with $w^{\prime}$.

Since the edges in $E_m$ induce a matching on $P_n$, $|\overline{V_m}|= n-2|E_m|=n-2k$. Recall from the proof of Lemma \ref{construction23} that in $w^{\prime}$, there exists no error between two vertices $v_1, v_2 \in \overline{V_m}$. Let us fix some $e^*_k\in E_m$. Using the proof of Lemma \ref{construction23}, we know that each vertex $v_i \in \overline{V_m}$ induces an error of $w^*_k$ with exactly one endpoint of $e^*_k$ (and no error with the other endpoint). Summing over all vertices $v_i \in \overline{V_m}$, we get that each edge $e^*_k \in E_m$ accumulates a total of $(n-2k)w^*_k$ in error. Summing again over all edges $e^*_k \in E_m$ yields the desired bound.

\end{proof}

\section{Graph Compression for Trees}
\label{treescompressed}
In this section, we study the problem of distance-preserving graph compression for weighted trees. Precisely, we study a relevant problem, referred to as \textit{the marking problem}, for a tree $T=(V,E)$, $|V|=n$, and weight function $w: E \xrightarrow[]{} \mathbb{R}_{\geq 0}$.

The remainder of this section is organized as follows. In Section \ref{sing}, we formally define the marking problem. The adaptation of the error function (Eq. (\ref{hehe})) to the marking problem is thoroughly explained in Section \ref{formul}. 
As a warm-up, we study a special case of the marking problem in Section \ref{equal}, after which we generalize the results in Section \ref{varying} and present a linear-time algorithm for solving the marking problem in Algorithm \ref{treesalg2}. As the final component of this section, we thoroughly study the difference between the marking problem (Definition \ref{marking}) and the fractional marking problem (Definition \ref{fracdef}) in Section \ref{fract}.
\subsection{The Marking Problem for a Single Edge}
\label{sing}
As seen in Section \ref{paths}, for merging a single edge in a weighted path, marking one of the neighbouring edges produces the optimal amount of error. An important question is how to generalize this result to solve the same problem for weighted trees. 
We formally state the marking problem as: 
\begin{definition}
    \textbf{The Marking Problem for Weighted Trees:} Given a contracted edge $e^*$ in a weighted tree $T$, what subset of the neighbouring edges of $e^*$ should we mark such that the error value of Eq. (\ref{hehe}) is minimized over all such possible subsets?
\label{marking}
\end{definition}
An example of the marking problem is depicted in Figure \ref{fig16}-(a), where edge $e^{*}$ with weight $w^{*}$ is contracted.  As shown in Figure \ref{fig16}-(b), in the marking problem, the goal is to mark a subset of the neighbouring edges of $e^*$, by setting the new weight of each marked edge $e_i$ to $w^{\prime}(e_i)=w(e_i) +\epsilon_i, \epsilon_i \in \{0, w^{*}\}$, in a way that minimizes the error function of Eq. (\ref{hehe}) over all such possible subsets. Note that the fractional case (when the weight of each marked edge $e_i$ is set to $w^{\prime}(e_i)=w(e_i) +\epsilon_i, \epsilon_i \in [0, w^{*}]$) is thoroughly studied in Section \ref{fract}.

In the tree of Figure \ref{fig16}-(a), $e^*$ has four neighbouring edges, namely $e_1=(v_1, v_3)$, $e_2=(v_1,v_4)$, $e_3=(v_2, v_5)$, and $e_4=(v_2, v_6)$. Different subsets of these neighbouring edges can be marked, for instance, in Figure \ref{fig17}-(a), $\{e_1, e_2\}$ is marked. In the remainder of this section, we may refer to each of these marked subsets as a \textit{marking} for simplicity. For example, in Figure \ref{fig17}-(c), $\{e_1, e_3\}$ is a marking. An optimal marking is one that minimizes the error function of Eq. (\ref{hehe}) over all possible markings.

Since for merging an edge in a weighted path marking one of the neighbouring edges gives the optimal amount of error, our intuition tells us that in a weighted tree, we have to mark all neighbouring edges on one side of the contracted edge $e^*$. As we shall show later, this intuition, though not completely correct, is optimal for specific kinds of input. To study the marking problem, we first present some definitions and observations using Figure \ref{fig16} and Figure \ref{fig17} as our running examples. We assume the tree is laid out in the plane and $e^*$ (the edge to be merged) is horizontal. This assumption will simplify the description of our results.
        \begin{figure*}[h]
    \centering
        \subfloat[Before merging $e^*$]{{\includegraphics[width=170pt]{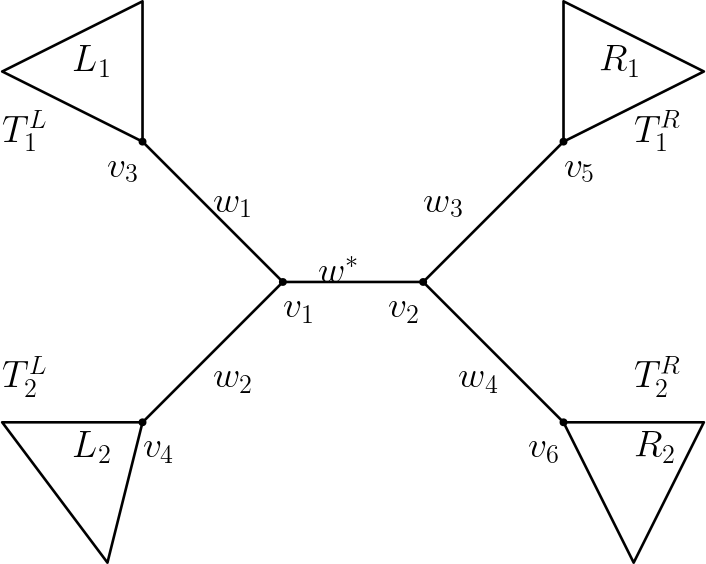}}}
        \hspace{2 cm}
        \subfloat[After merging $e^*$]{{\includegraphics[width=150pt]{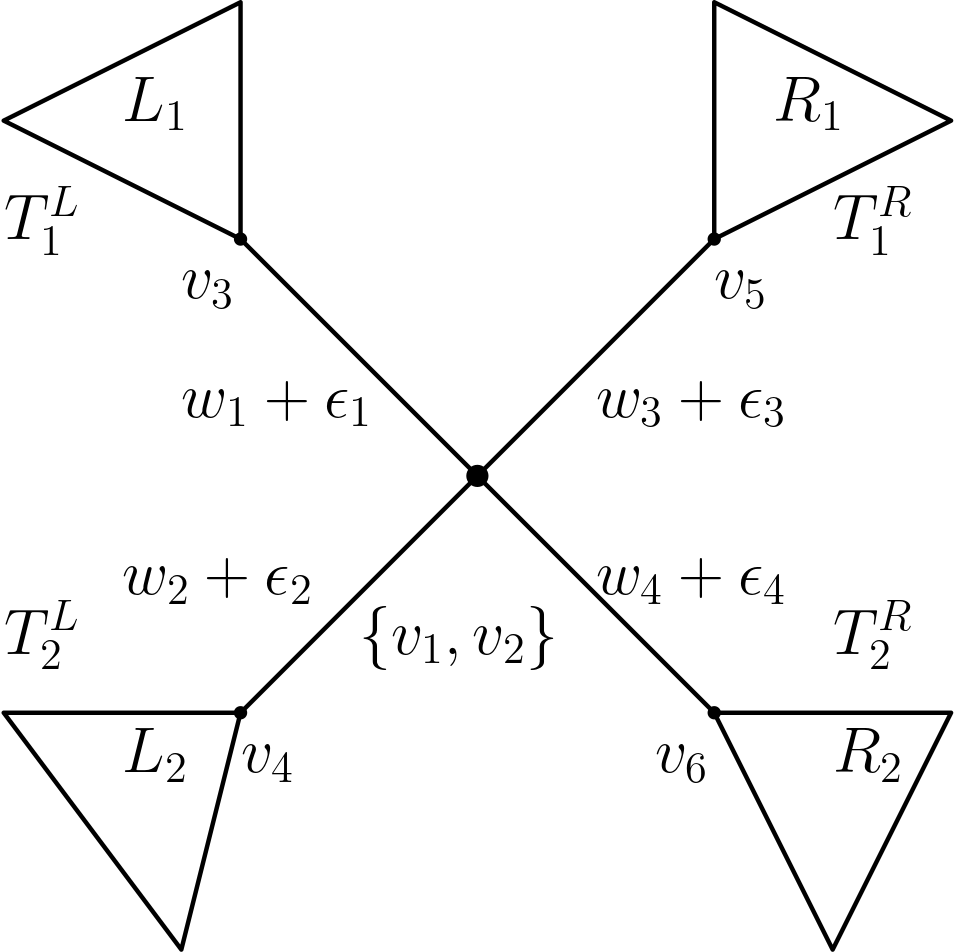}}}\

                \qquad \hspace{2 cm}
    \qquad
    \caption{The figure used in Section \ref{sing} for defining the marking problem. We denote by $\treeleft, i \in \{1,2 \}$, and $\treeright, j \in \{1,2\}$ the subtree rooted at the $i$-th edge to the left and the $j$-th edge to the right respectively. Moreover, $L_i=|\{v|v\in \treeleft\}| $ and  $R_j=|\{v|v\in \treeright\}| $ denote the number of vertices in each subtree. }
    \label{fig16}
\end{figure*}
\begin{definition}
    Let $T=(V, E) $ be a weighted tree with non-negative weights, and let $e^*=(v_1, v_2)$ be the merged edge with weight $w^*$, $V_m=\{v_1, v_2\}$, and $\overline{V_m}=V-V_m$. We denote by $\leffs$ the number of subtrees to the left of $v_1$ and by $\riis$ the number of subtrees to the right of $v_2$. More formally, let $E^{\prime}=E- e^*$. We have: 
    \begin{linenomath*}$$V_L= \{u| (u,v_1) \in E^{\prime}\}, \leffs =|V_L|$$\end{linenomath*}
    \begin{linenomath*}$$V_R= \{w| (v_2,w) \in E^{\prime}\}, \riis= |V_R|$$\end{linenomath*}
    \label{weightedtree}
\end{definition}
For instance, in the tree of Figure \ref{fig16}, we have $V_L=\{v_3, v_4\}$ and $V_R=\{v_5, v_6\}$ and therefore $\leffs=\riis=2$.\\
Given $e^*=(v_1, v_2)$ in $T$, $T - \{v_1, v_2\}$ is a forest $\mathcal{F}$, the components of which are used in our analyses and defined as follows:
\begin{definition}
 Let $T$, $e^*=(v_1, v_2)$, $V_L$ and $V_R$ be as defined in Definition \ref{weightedtree}. Let $\mathcal{F}$ be the forest $T- \{v_1, v_2\}$. Furthermore, assume that the connected components of $\mathcal{F}$ are rooted at the vertices of $V_L$ or $V_R$, and let $C_L$ and $C_R$ be the sets of components of $\mathcal{F}$ rooted at the vertices of $V_L$ and $V_R$ respectively. Then, we denote by $\treeleft, i \in \{1, \dots, \leffs \}$ the $i$-th member of $C_L$, and by $\treeright, j \in \{1,\dots, \riis\}$ the $j$-th member of $C_R$, given some arbitrary ordering on the members of $C_L$ and $C_R$.
 \label{weightedtree2}
\end{definition}
In the tree of Figure \ref{fig16}, $\leffs=2$, and $C_L$ has two members (the subtrees rooted at $v_3$ and $v_4$). Given some arbitrary ordering on the members of $C_L$, $T^L_1$ is the subtree rooted at $v_3$. 
\\
We also formally define the cardinality of the subtrees of Definition \ref{weightedtree2} as follows: 
\begin{definition}
    Let $\treeleft, i \in \{1, \dots, \leffs \}$ and  $\treeright, j \in \{1,\dots, \riis\}$ be as defined in Definition \ref{weightedtree2}. We have $L_i=|\{v|v\in \treeleft\}| $ and  $R_j=|\{v|v\in \treeright\}|$. We refer to $L_i$ as the cardinality of the $i$-th edge on the left and $R_j$ as the cardinality of the $j$-th edge on the right.
\end{definition}
A few examples of marking the edges of Figure \ref{fig16} are provided in Figure \ref{fig17}-(a) to Figure \ref{fig17}-(c). In Figure \ref{fig17}-(a) and Figure \ref{fig17}-(b), all edges on one side of $e^*$ are marked, and in Figure \ref{fig17}-(c), a subset of edges from both sides is marked. Marking an edge could both increase and decrease the total amount of error. Before proceeding with the remainder of this section, we note the following lemma to justify our focus on minimizing the error between all pairs of vertices in $\overline{V_m}$.

\begin{lemma}
(See Figure \ref{fig16}) Let $e^*=(v_1,v_2)$ be the single merged edge in a weighted tree $T=(V, E)$, and let $\overline{V_m}= V- \{v_1, v_2\}$. Then, as long as every neighbouring edge of $e^{*}$ is either marked or unmarked, the error between some vertex $u \in \overline{V_m}$  and the vertices in $\{v_1, v_2\}$ is minimized.  
\label{depressed}
\end{lemma}
\begin{proof}
    This lemma is a direct result of Lemma \ref{l12} and Theorem \ref{l3aval}. Let us fix some vertex $u \in T_2^{L}$ (see Figure \ref{fig16}-(b)), the error between $u$ and the endpoints of $e^*$, $v_1$ and $v_2$, can be formulated as:
    \begin{linenomath*}\begin{align*}
        |\Delta E|^{\prime}=&\underbrace{|w_2-(w_2+\epsilon_2)|}_{\text{between } u \text{ and }v_1}+\underbrace{|w_2+w^*-(w_2+\epsilon_2)|}_{\text{between } u \text{ and }v_2 }=|\epsilon_2|+|w^*-\epsilon_2|= |\epsilon_2|+|\epsilon_2- w^*|
    \end{align*}\end{linenomath*}
    Using Lemma \ref{l12}, we have $|\Delta E|^{\prime} \geq w^*$, and $|\Delta E|^{\prime} = w^*$ for $0\leq \epsilon_2 \leq w^*$. Therefore, when $(v_1, v_4)$ is either marked or unmarked, we have $\epsilon_2 \in \{0, w^*\}$, which satisfies the desired conditions. This analysis applies to all nodes $u \in \overline{V_m}$, thus the lemma follows.  
\end{proof}
  
In the remainder of this section, we therefore only focus on minimizing the error between all pairs of vertices $u_1, u_2 \in \overline{V_m}$, because by the definition of the marking problem (Definition \ref{marking}), the conditions of Lemma \ref{depressed} are automatically satisfied. 
\subsection{Formulating The Error}
\label{formul}
        \begin{figure*}[h]
    \centering
        \subfloat[Both edges on the left are marked.]{{\includegraphics[width=150pt]{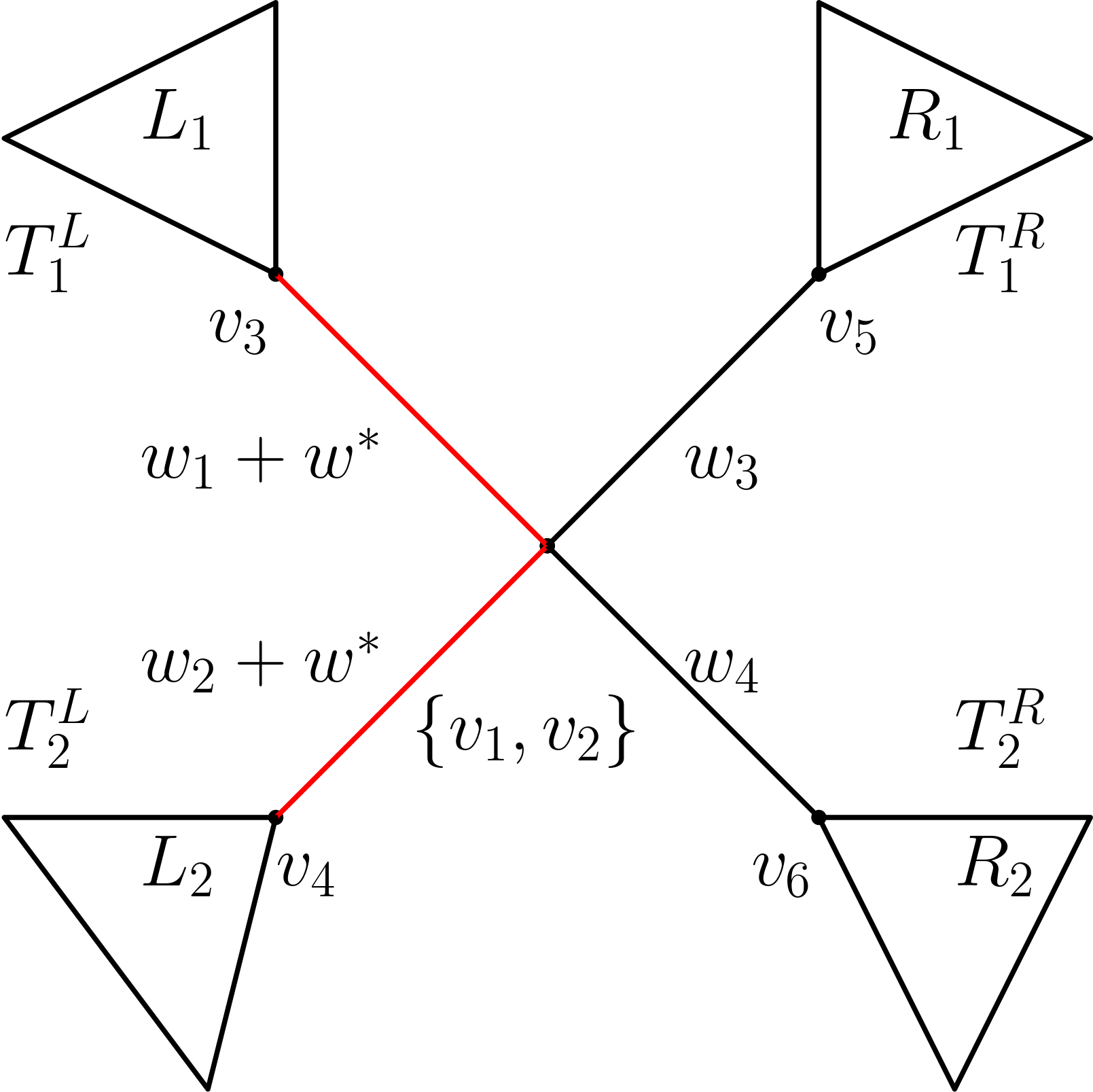}}}
        \hspace{2 cm}
        \subfloat[Both edges on the right are marked. ]{{\includegraphics[width=150pt]{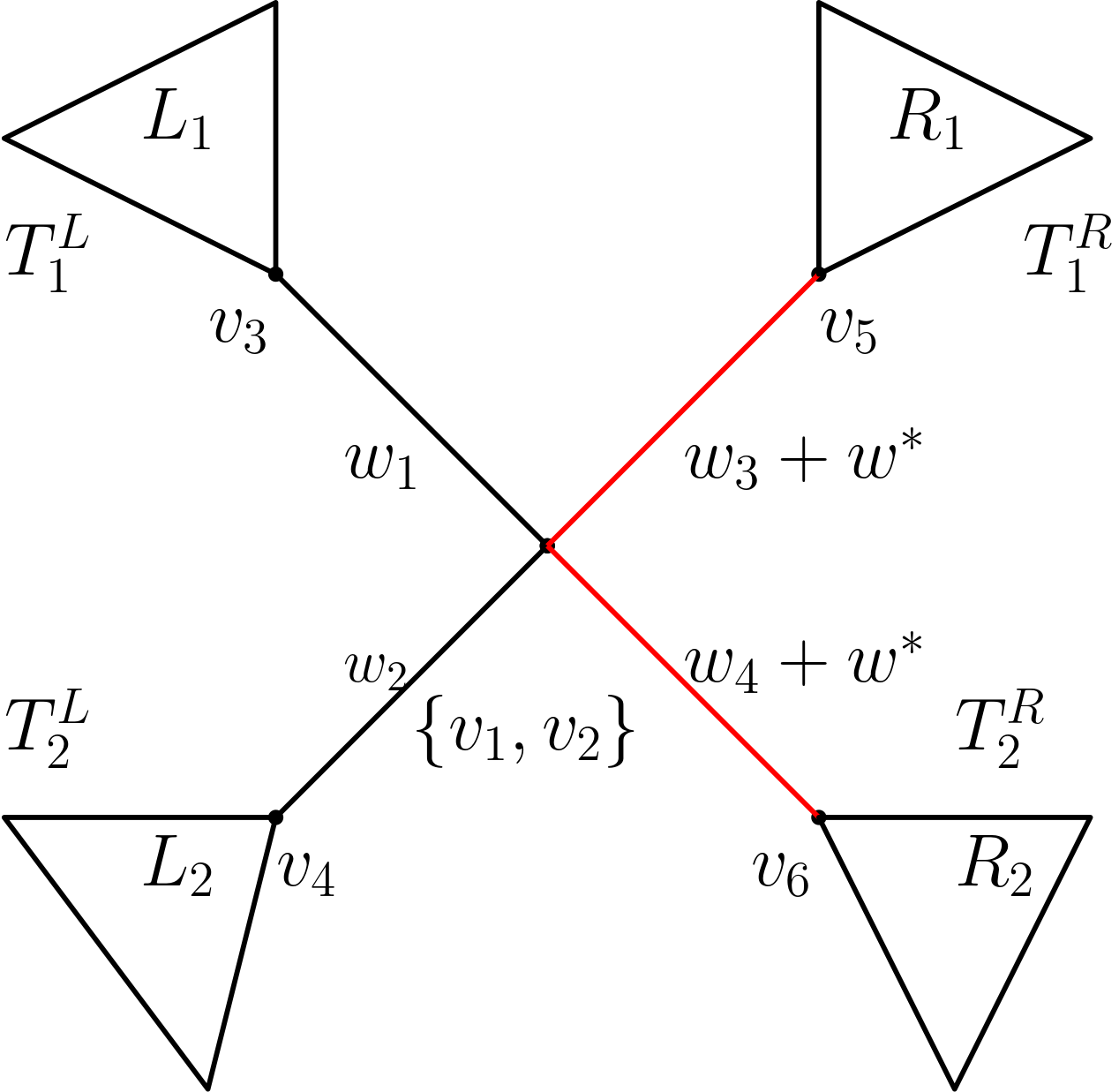}}}\

         \hspace{2 cm}\subfloat[A subset of edges from both sides is marked.]{{\includegraphics[width=150pt]{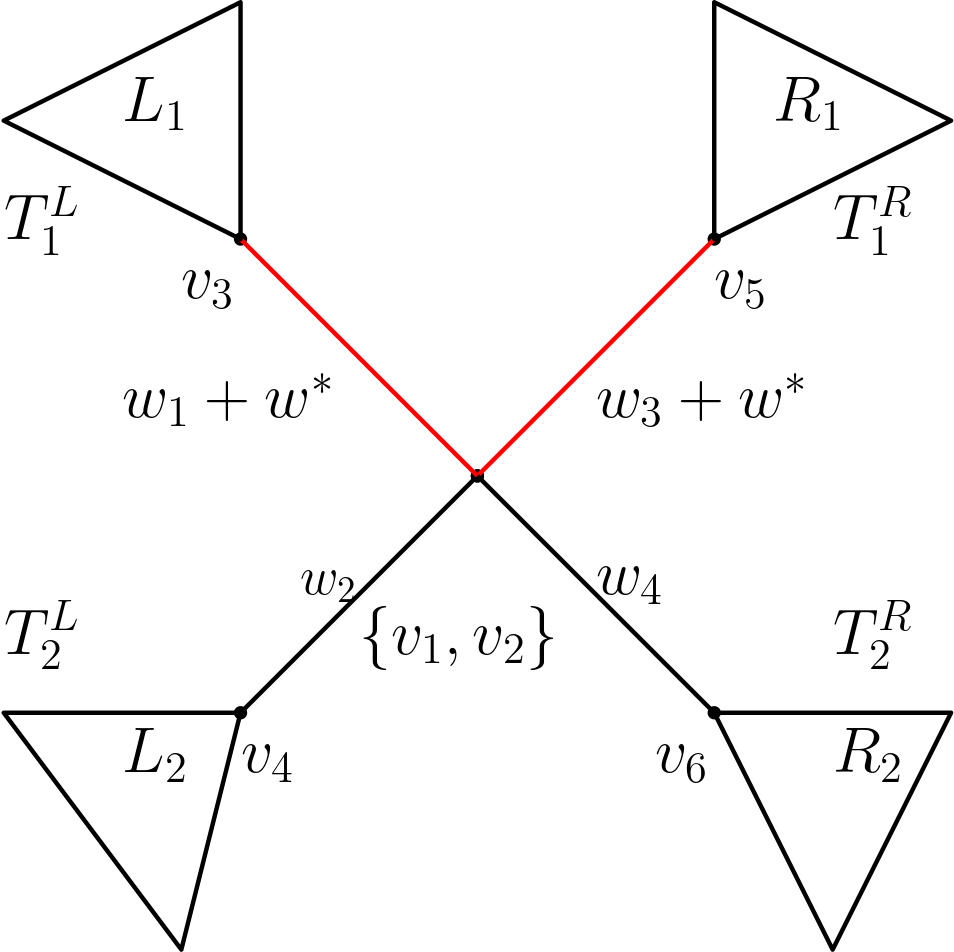}}}
                \qquad \hspace{2 cm}
    \qquad
    \caption{The figure used in Section \ref{formul} for formulating the marking error, the marked edges are highlighted in red. We denote by $\treeleft, i \in \{1,2 \}$, and $\treeright, j \in \{1,2\}$ the subtree rooted at the $i$-th edge to the left and the $j$-th edge to the right respectively. Moreover, $L_i=|\{v|v\in \treeleft\}| $ and  $R_j=|\{v|v\in \treeright\}| $ denote the number of vertices in each subtree.  }
    \label{fig17}
\end{figure*}
This section formally explains how marking a set of edges affects the error function. Using Figure \ref{fig17}, we first present some examples, which we generalize later in Observation \ref{obi}. Throughout this section, we may sometimes refer to this error as \textit{units of error}, where each unit is equal to $w^*$. 
 \begin{example}
 The error between $v_3$ and $v_4$ in Figure \ref{fig17}-(a) is equal to $|w_1 +w^* +w_2 +w^*- w_1 -w_2|=2w^{*}$. In the original graph (Figure \ref{fig16}-(a)), $e^{*}$ does not appear on the unique path between $v_3$ and $v_4$, while in the modified graph (Figure \ref{fig17}-(a)), the weight of $e^*$ appears twice. In the marking of Figure \ref{fig17}-(a), the total amount of error between all pairs of vertices $u_1 \in T_1^L, u_2 \in T_2^L$ is $L_1 \times L_2 \times 2w^*$. 
 \label{obs1}
 \end{example}

 \begin{example}
The error between $v_3$ and $v_5$ in Figure \ref{fig17}-(c) is $|w_1+w^*+w_3+w^*-w_1-w^*-w_3|=w^*$. Because in the original graph (Figure \ref{fig16}-(a)), $e^*$ appears only once on the unique path from $v_3$ to $v_5$, while in the modified graph (Figure \ref{fig17}-(c)), the weight of $e^*$ appears twice. The total amount of error between all pairs of vertices $u_1 \in T_1^L, u_2 \in T_1^R$ is $L_1 \times R_1 \times w^*$.
\label{obs3}
 \end{example}
 \begin{example}
In Figure \ref{fig17}-(c), the error between $v_5$ and $v_6$ is $|w_3+w^*+w_4-w_3-w_4|=w^*$. The total amount of error between all pairs of vertices $u_1 \in T_1^R, u_2 \in T_2^R$ is $R_1 \times R_2 \times w^*$. 
 \label{obs4}
 \end{example}
\begin{example}
In Figure \ref{fig17}-(c), the total amount of error between all pairs of vertices $u_1 \in T_1^L, u_2 \in T_2^L$ is $L_1 \times L_2 \times w^*$.
 \label{obs5}
\end{example}
\begin{example}
     In Figure \ref{fig17}-(a), the error between $v_3$ and $v_5$ is equal to $|w_1+w^*+w_3-w_1-w^*- w_3|=0$. The length of the unique path between $v_3$ and $v_5$ does not change compared with Figure \ref{fig16}-(a).
      \label{obs6}
\end{example}

\begin{observation}
 Between the vertices of two edges (vertices belonging to the subtree rooted at that edge) adjacent to the endpoints of $e^*$, there might exist some error. We classify this observation into the following cases:
 \label{obi}
\end{observation}
\begin{enumerate}
    \item Let $T_i^L$ and $T_j^L$ be the subtrees adjacent to two distinct marked edges on the left. Then, the total amount of error between all pairs of vertices $u_1 \in T_i^L, u_2 \in T_j^L$ is $L_i \times L_j \times 2w^*$ (see Example \ref{obs1}).
   \item Let $T_i^R$ and $T_j^R$ be the subtrees adjacent to two marked edges on the right. Then, the total amount of error between all pairs of vertices $u_1 \in T_i^R, u_2 \in T_j^R$ is $R_i \times R_j \times 2w^*$.
    \item Let $T_i^L$ and $T_j^R$ be the subtrees adjacent to two marked edges on the left and right respectively. Then, the total amount of error between all pairs of vertices $u_1 \in T_i^L, u_2 \in T_j^R$ is $L_i \times R_j \times w^*$ (see Example \ref{obs3}).
    \item Let $T_i^R$ and $T_j ^R$ be the subtrees adjacent to a marked edge and an unmarked edge on the right respectively. Then, the total amount of error between all pairs of vertices $u_1 \in T_i^R, u_2 \in T_j^R$ is $R_i \times R_j \times w^*$ (see Example \ref{obs4}).
    \item Let $T_i^L$ and $T_j ^L$ be the subtrees adjacent to a marked edge and an unmarked edge on the left respectively. Then, the total amount of error between all pairs of vertices $u_1 \in T_i^L, u_2 \in T_j^L$ is $L_i \times L_j \times w^*$ (see Example \ref{obs5}).
    \item Let $T_i^L$ be the subtree adjacent to a marked edge on the left, and $T_j^R$ be the subtree adjacent to an unmarked edge on the right. Then, the total amount of error between all pairs of vertices $u_1 \in T_i^L, u_2 \in T_j^R$ is equal to zero (see Example \ref{obs6}).
    \item Let $T_i^L$ be the subtree adjacent to an unmarked edge on the left, and $T_j^R$ be the subtree adjacent to a marked edge on the right. Then, the total amount of error between all pairs of vertices $u_1 \in T_i^L, u_2 \in T_j^R$ is equal to zero.
\end{enumerate} 
\subsection{Equal-Sized Subtrees}
\label{equal}
 We now investigate a special case where each subtree on the left has $\nl$ vertices and each subtree on the right has $\nr$ vertices, i.e. $L_i = n_L, \; 1\leq i \leq \leffs$, and $R_i=\nr, \; 1\leq i \leq \riis$. Recall that every merged edge has two sides, left and right, one of which is designated as the \textit{preferable side}. A given side is preferable if it produces a smaller amount of error when fully marked compared to its fully-marked counterpart. For example, if the left side is preferable, we have: 
\begin{linenomath*}\begin{equation}
\label{prefer}
\nll \times \leffs (\leffs -1)  \leq \nrr \times \riis (\riis -1)
\end{equation}\end{linenomath*}
The above inequality compares the error between the marking with the left side fully marked and the right side fully unmarked (Figure \ref{fig17}-(a)), and the opposite marking with the right side fully marked and the left side fully unmarked (Figure \ref{fig17}-(b)). In the first marking, there exists no error  between the left and the right sides (Observation \ref{obs}, Case 6), but there are $\leffs \choose 2$ distinct pairs of marked edges on the left, each inducing an error of $\nl \times \nl \times 2w^*$ (Observation \ref{obs}, Case 1). Therefore, the total amount of error for the first marking is equal to ${\leffs \choose {2}} \times n_L \times n_L \times 2w^*= \nll \times \leffs(\leffs -1 ) \times w^*$. The other marking can be analyzed analogously. Note that in the remainder of this section, we drop $w^*$ from each error term, and each error term counts the error units, where each unit is equal to $w^*$. Therefore, all quantities are implicitly multiplied by $w^*$ in the remainder of this section.

The following lemma states that, for a contracted edge $e^*$ that has equal-sized subtrees on each side, the optimal solution is obtained by marking all edges on the preferable side of $e^*$ and leaving the other side completely unmarked.
\begin{lemma}
\label{aymadareto}
Given a merged edge $\merged$ (in a weighted tree) with two sides left and right, such that the subtrees on each side have equal sizes, the optimal marking is obtained if one side (the preferable side) is fully marked and the other side is fully unmarked.
\end{lemma}
\begin{proof}
By contradiction. This lemma assumes each subtree on the left and right side has $\nl$ and $\nr$ vertices respectively, i.e. $L_i = n_L, \; 1\leq i \leq \leffs$, and $R_i=\nr, \; 1\leq i \leq \riis$. Without loss of generality, we assume the left side is preferable throughout this proof. Therefore, we have: 
\begin{linenomath*}$$\nll \times \leffs (\leffs -1)  \leq \nrr \times \riis (\riis -1)$$\end{linenomath*}
Let $i$ and $j$ denote the number of marked edges on the left and right respectively. We define two functions, $\text{MARK\_LEFT}$, which marks one of the edges on the left, and $\text{UNMARK\_RIGHT}$, which unmarks one edge on the right. We will show that for all values $i< \leffs$ or $j>0$, one can achieve smaller error values by applying a series of $\text{MARK\_LEFT}$'s and $\text{UNMARK\_RIGHT}$'s and ending up at $i=\leffs$ and $j=0$, as desired. For a function $f \in \mathcal{F}= \{\text{MARK\_LEFT},\text{UNMARK\_RIGHT}\}$, we define $\Delta(f)$ as the amount of change in the error value after applying $f$ to the tree. Since we are interested in decreasing the error value using the functions in $\mathcal{F}$, in this proof we will look for conditions under which $\Delta(\text{MARK\_LEFT}) \leq 0$ and $\Delta(\text{UNMARK\_RIGHT}) \leq 0$.

We begin by investigating $\mle$. Note that this function sets $i\xleftarrow{} i+1$ and $j \xleftarrow{} j$. We observe the following:
\begin{enumerate}
    \item Because we are marking a new edge, the total amount of error between the marked edges on the left changes by: 
    \begin{linenomath*}$$\nll\times 2 \bigg({i+1 \choose 2} -{i \choose 2} \bigg)= \nll \times 2i$$\end{linenomath*}
    \item The total amount of error between the unmarked edges and the marked ones on the left changes by:
    \begin{linenomath*}$$\nll\times ((i+1)(\leffs-i-1)-i(\leffs-i))=\nll\times (i\leffs-i^2-i+\leffs-i-1-i\leffs+i^2)=\nll\times(\leffs -2i-1)$$\end{linenomath*}
    \item The total amount of error between the marked edges on the left and right changes by:
    \begin{linenomath*}$$\nl\nr((i+1)j-ij)=\nl\nr \times j$$\end{linenomath*}
    \item The total amount of error between the unmarked edges on the left and right changes by: 
    \begin{linenomath*}$$\nl \nr \times ((\leffs-i-1)(\riis -j)-(\leffs-i)(\riis-j))=\nl \nr \times (\leffs \riis -\leffs j-i \riis +ij-\riis+j-\leffs \riis +\leffs j+i \riis -ij)=\nl \nr \times (j-\riis)$$\end{linenomath*}

\end{enumerate}
Therefore, $\Delta (\mle)$ is equal to: 
\begin{linenomath*}$$\Delta (\mle)=\nll\times (2i + \leffs -2i -1)+ \nl \nr (j+j-\riis)= \nll\times (\leffs  -1)+ \nl \nr (2j-\riis)$$\end{linenomath*}
Since we are looking for conditions under which $\Delta (\mle) \leq 0$, we have: 
\begin{linenomath*}$$
\begin{array}{cc}
    \Delta (\mle)\leq 0 \xrightarrow[]{} \nll\times (\leffs  -1)+ \nl \nr (2j-\riis) &\leq 0\\
    
\end{array}
$$\end{linenomath*}
Therefore, 
\begin{linenomath*}\begin{equation}
    \Delta (\mle)\leq 0 \text{ if } \nl\times (\leffs  -1) \leq \nr (\riis-2j) \xrightarrow[]{\text{Rearranging the terms} } j  \leq \frac{ \riis }{2}+ \frac{\nl(1-\leffs)}{2 \nr} 
    \label{mark}
\end{equation}\end{linenomath*}
Similar reasoning can be used for $\unmr$. This function sets $i\xleftarrow{} i$ and $j \xleftarrow{} j-1$. We have:
\begin{enumerate}
    \item The total amount of error between the marked edges on the right changes by: 
    \begin{linenomath*}$$\nrr\times 2 \big({j-1 \choose 2} -{j \choose 2} \big)= \nrr \times (-2(j-1))$$\end{linenomath*}
    \item The total amount of error between the unmarked edges and the marked ones on the right changes by:
    \begin{linenomath*}$$\nrr\times ((j-1)(\riis -j +1)-j(\riis-j))=\nrr\times (j\riis-j^2+j-\riis+j-1-j\riis +j^2)=\nrr\times(2j-\riis-1)$$\end{linenomath*}
    \item The total amount of error between the marked edges on the left and right changes by:
    \begin{linenomath*}$$\nl\nr(i(j-1)-ij)=\nl\nr \times (-i)$$\end{linenomath*}
    \item The total amount of error between the unmarked edges on the left and right changes by: 
    \begin{linenomath*}$$\nl \nr \times ((\leffs-i)(\riis -j+1)-(\leffs-i)(\riis-j))=\nl \nr \times (\leffs \riis -\leffs j+\leffs- i\riis +ij-i-\leffs \riis +\leffs j+i \riis -ij)=\nl \nr \times (\leffs -i )$$\end{linenomath*}
\end{enumerate}
Thus, we have: 
\begin{linenomath*}$$\Delta(\unmr)=\nrr(-2(j-1)+2j-\riis-1)+\nl \nr(\leffs -i -i )=\nrr(1-\riis)+\nl \nr(\leffs -2i )$$\end{linenomath*}
and
\begin{linenomath*}\begin{equation}
    \Delta (\unmr)\leq 0 \text{ if }  \nl (\leffs -2i )\leq \nr(\riis-1) \xrightarrow[]{\text{Rearranging the terms} } i  \geq \frac{ \leffs}{2}+ \frac{\nr(1-\riis)}{2 \nl} 
    \label{unmark}
\end{equation}\end{linenomath*}
We conclude the proof by stating that whenever $i< \leffs$ or $j>0$, one can achieve smaller error values by applying a series of $\text{MARK\_LEFT}$'s and $\text{UNMARK\_RIGHT}$'s and ending up at $i=\leffs$ and $j=0$. When $j  \leq \frac{ \riis }{2}+ \frac{\nl(1-\leffs)}{2 \nr} $, Eq. (\ref{mark}) is satisfied.  Therefore, we repeatedly apply $\mle$ until $i=\leffs$, at which point Eq. (\ref{unmark}) is satisfied and we repeatedly apply $\unmr$ until $j=0$, as desired.
Now suppose $j  > \frac{ \riis }{2}+ \frac{\nl(1-\leffs)}{2 \nr}$  edges are marked on the right side. If $i  \geq \frac{ \leffs}{2}+ \frac{\nr(1-\riis)}{2 \nl}$, Eq. (\ref{unmark}) is satisfied, which allows us to repeatedly apply $\unmr$ until $j=0$, at which point Eq. (\ref{mark}) is satisfied and we repeatedly apply $\mle$ until $i=\leffs$, as desired.

Assume $i  < \frac{ \leffs}{2}+ \frac{\nr(1-\riis)}{2 \nl}$ and $j  > \frac{ \riis }{2}+ \frac{\nl(1-\leffs)}{2 \nr} $, and both Eq. (\ref{mark}) and Eq. (\ref{unmark}) are unsatisfied. We first apply $\frac{ \leffs}{2}+ \frac{\nr(1-\riis)}{2 \nl}  -i= \frac{\leffs \nl +\nr(1-\riis)}{2 \nl}  -i$ $\mle$'s, increasing the error by $(\frac{\leffs \nl +\nr(1-\riis)}{2 \nl}  -i) (\nll\times (\leffs  -1)+ \nl \nr (2j-\riis))$, at which point $i=\frac{\leffs \nl +\nr(1-\riis)}{2 \nl} 
$ and $\Delta(\unmr)=0$. Therefore, we set $j\xleftarrow[]{}0$  without changing the error (since $\Delta(\unmr)=0$), and then we apply $\leffs - (\frac{\leffs \nl +\nr(1-\riis)}{2 \nl})=\frac{\leffs \nl -\nr(1-\riis)}{2 \nl}$ $\mle$'s until $i=\leffs$, as desired. We now show that this sequence of $\mle$'s and $\unmr$'s results in an error value no worse than that of the original one:

\begin{linenomath*}\begin{align*}
    \Delta =& (\underbrace{\frac{\leffs \nl +\nr(1-\riis)}{2 \nl}  -i}_{\leq \frac{\leffs \nl +\nr(1-\riis)}{2 \nl} }) (\underbrace{\nll\times (\leffs  -1)+ \nl \nr (2j-\riis)}_{\leq \nll\times (\leffs  -1)+ \nl \nr (\riis )})+ j \underbrace{(\nrr(1-\riis)+\nl \nr(\leffs -2 \times \frac{\leffs \nl +\nr(1-\riis)}{2 \nl} )}_{=0})  \\
   + &\frac{\leffs \nl -\nr(1-\riis)}{2 \nl} (\nll\times (\leffs  -1)+ \nl \nr (-\riis))   \\ \leq& (\frac{\leffs \nl +\nr(1-\riis)}{2 \nl})( \nll\times (\leffs  -1)+ \nl \nr (\riis )) + \frac{\leffs \nl -\nr(1-\riis)}{2 \nl} (\nll\times (\leffs  -1)+ \nl \nr (-\riis))  \\
    =&\nll \times \leffs (\leffs -1) - \nrr \times \riis (\riis -1)\\
    \leq &0  
\end{align*}\end{linenomath*}
and we arrive at $i=\leffs$ and $ j=0$ while obtaining a smaller error value.
\end{proof}
In the next section, we generalize Lemma \ref{aymadareto} to the case in which different subtrees can have varying sizes. 
\subsection{Varying-Size Subtrees}
\label{varying}

      \begin{figure*}[h]
    \centering
        \subfloat[A full marking of the edges on the left with error count 32.]{{\includegraphics[width=140pt]{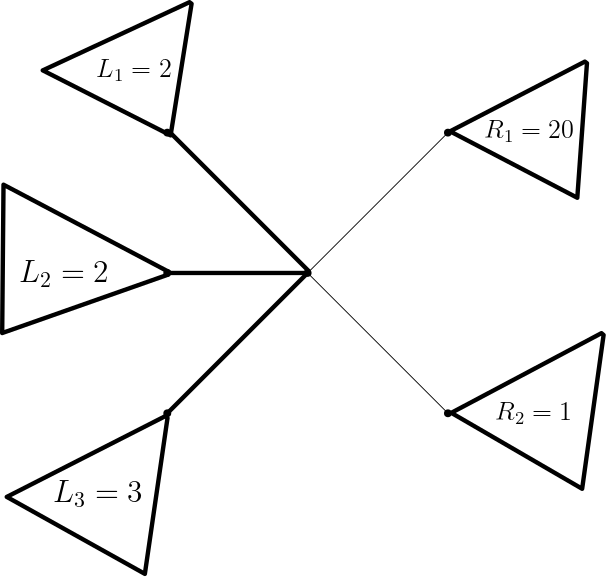}}}
        \hspace{2 cm}
        \subfloat[A full marking of the edges on the right with error count 40.]{{\includegraphics[width=140pt]{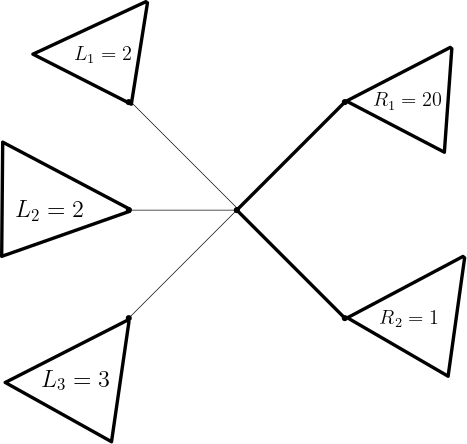}}}\
        \qquad \hspace{2 cm}
        \subfloat[A marking with edges marked on both sides with error count 75.]{{\includegraphics[width=140pt]{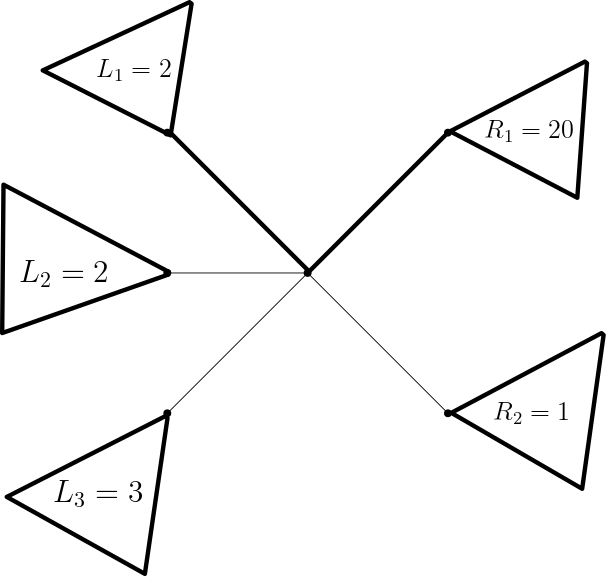}}}
                \qquad \hspace{2 cm}
        \subfloat[The optimal solution with error count 27.]{{\includegraphics[width=140pt]{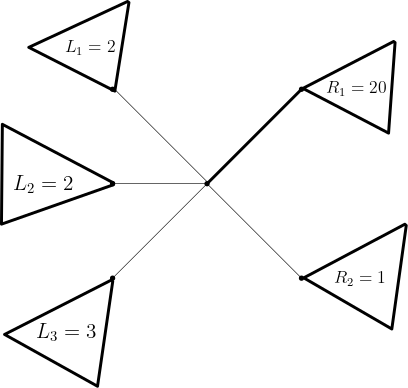}}}
    \qquad
    \caption{An example of tree compression in which the edges on each side have different-sized subtrees. The optimal solution does not have a full marking on any side. However, the optimal solution only has marked edges on one side (this is always the case as shown in Lemma \ref{lemma11}). }
    \label{fig6}
\end{figure*}
As a generalization of Section \ref{equal}, now assume the $i$-th subtree on the left ($1 \leq i \leq \leffs$) has $\leff_i$ nodes, and the $j$-th subtree on the right ($1 \leq j \leq \riis $) is of size $\rii_j$. We observe when each side has subtrees of different sizes, marking all edges on one side does not necessarily produce the optimal error. An example is depicted in Figure \ref{fig6}, where marking only one edge on the right produces the optimal amount of error. 
Although marking all edges on one side does not necessarily produce the optimal error, we observe that no optimal solution has markings on both sides, as the following lemma states. Similar to Section \ref{equal}, we remove $w^*$ from all calculations and expressions in this section. Therefore, all calculations in this section are implicitly multiplied by $w^*$.
\begin{lemma}
\label{lemma11}
Given a merged edge $\merged$ (in a weighted tree) with two sides left and right, no optimal marking has marked edges on both sides.

\end{lemma}
\begin{proof}
    By contradiction. We assume there exists such an optimal marking, and we strictly improve its error by unmarking everything on one of the two sides (thus obtaining a contradiction). Expanding the proof of Lemma \ref{aymadareto}, we define four operations, $\mle$, $\umle$, $\mr$, and $\unmr$ for unmarking and marking edges on both ends. For a function \begin{linenomath*}$$f \in \mathcal{F}= \{\mle, \umle, \mr,  \unmr \}$$\end{linenomath*} we define $\Delta(f)$ as the amount of change in the error value after applying $f$ to the tree. Let $\sumleft =\sum_{i=1}^{\leffs}L_i$ and $\sumright =\sum_{i=1}^{\riis}R_i$ denote the total sum of all edge cardinalities on the left and right sides respectively. Furthermore, let $\sumleftm$, $\sumleftu$, $\sumrightm$, and $\sumrightu$ denote the sum of the cardinalities of the marked and unmarked edges on the left and right sides respectively. Note that $\sumleft= \sumleftu + \sumleftm$ and $\sumright=\sumrightu + \sumrightm$.

    First, we calculate $\Delta(\unmr)$ and derive $\Delta(\umle)$ by symmetry. Assume we are unmarking the $i$-th edge on the right $e_i$ with cardinality $R_i$. We break the change in the error value down into four parts as follows: 
    \begin{enumerate}
    \item The total amount of error between the marked edges on the right changes by: 
    \begin{linenomath*}$$-2 \times R_i \times (\sumrightm - R_i)$$\end{linenomath*}
    because between two marked edges on the right, there exist two units of error (equal to twice the weight of the merged edge $\merged$). Therefore, unmarking $e_i$ relieves some of this error.
    \item The total amount of error between the unmarked and the marked edges on the right changes by:
    \begin{linenomath*}$$-R_i \times (\sumrightu) + R_i \times (\sumrightm -R_i)$$\end{linenomath*}
    because between a marked and an unmarked edge on the right, there exists one unit of error and unmarking $e_i$ relieves some error with other unmarked edges (the first part of the expression), making $e_i$ an unmarked edge itself (the second part of the expression).
    \item The total amount of error between $e_i$ and the marked edges on the left changes by:
    \begin{linenomath*}$$-R_i \times \sumleftm$$\end{linenomath*}
    because between two marked edges on the right and the left, there exists one unit of error and unmarking $e_i$ relieves some of this error.
    \item The total amount of error between $e_i$ and the unmarked edges on the left changes by: 
    \begin{linenomath*}$$R_i \times \sumleftu$$\end{linenomath*}
\end{enumerate}
Summing all four parts together, we get: 
\begin{linenomath*}\begin{equation}
\Delta(\unmr)= R_i \times \bigg(- (\sumrightm - R_i) - \sumrightu -  \sumleftm + \sumleftu \bigg) 
\label{delta1}
\end{equation}\end{linenomath*}
By symmetry, we also have: 
\begin{linenomath*}\begin{equation}
\Delta(\umle)= L_i \times \bigg(- (\sumleftm - L_i) - \sumleftu -  \sumrightm + \sumrightu \bigg) 
\label{delta2}
\end{equation}\end{linenomath*}
Next, we calculate $\Delta(\mle)$ and derive $\Delta(\mr)$ by symmetry. Assume we are marking the $i$-th edge on the left $e_i$ with cardinality $L_i$. We break the change in the error value into four parts: 
            \begin{enumerate}
    \item The total amount of error between the marked edges on the left changes by: 
    \begin{linenomath*}$$2 \times L_i \times (\sumleftm)$$\end{linenomath*}
    because between two marked edges on the left, there exist two units of error (equal to twice the weight of the merged edge $\merged$). Therefore, marking $e_i$ introduces some error between $e_i$ and all other marked edges on the left. 
    \item The total amount of error between the unmarked and the marked edges on the left changes by:
    \begin{linenomath*}$$-L_i \times (\sumleftm) + L_i \times (\sumleftu -L_i)$$\end{linenomath*}
    because between a marked and an unmarked edge on the left, there exists one unit of error, and marking $e_i$ relieves some error with all other marked edges on the left (the first part of the expression), making $e_i$ a marked edge itself (the second part of the expression).
    \item The total amount of error between $e_i$ and the marked edges on the right changes by:
    \begin{linenomath*}$$+L_i \times \sumrightm$$\end{linenomath*}
    because between two marked edges on the right and the left, there exists one unit of error, thus marking $e_i$ introduces some error between $e_i$ and all other marked edges on the right.
    \item The total amount of error between $e_i$ and the unmarked edges on the right changes by: 
    \begin{linenomath*}$$-L_i \times \sumrightu$$\end{linenomath*}
    because between two unmarked edges on the left and the right, there exists one unit of error, and marking $e_i$ relieves some of this error. 
\end{enumerate}
Summing all four parts together, we get: 
\begin{linenomath*}\begin{equation}
\Delta(\mle)= L_i \times \bigg(\sumleftm + (\sumleftu -L_i)+ \sumrightm - \sumrightu\bigg) 
\label{delta3}
\end{equation}\end{linenomath*}
By symmetry, we also have: 
\begin{linenomath*}\begin{equation}
\Delta(\mr)=  R_i \times \bigg(\sumrightm + (\sumrightu -R_i)+ \sumleftm - \sumleftu\bigg) 
\label{delta4}
\end{equation}\end{linenomath*}
Now, we can complete the proof. For the sake of contradiction, assume that there exists an optimal marking $\optimalmarking$ 
with edges marked on both sides. Therefore, we have $\sumleftm 
>0$ and $\sumrightm >0 $. Without loss of generality, assume 
$\sumrightu \geq \sumleftu$ (see Figure \ref{fig7}-(a)). Using 
Eq. (\ref{delta1}), we can unmark any edge on the right, say the 
$i$-th one $e_i$ connected to $R_i$ vertices, such that the 
change in the error value is equal to: 
\begin{linenomath*}$$
\begin{array}{cc}
    \Delta &=R_i \times \bigg(\underbrace{- (\sumrightm - R_i)}_{<0} - \sumrightu 
    \underbrace{-  \sumleftm}_{<0} + \sumleftu \bigg)
   <R_i (\underbrace{-\sumrightu +\sumleftu}_{\leq 0} )< 0
\end{array}
$$\end{linenomath*}
and we obtain a strictly better 
marking by unmarking $e_i$; 
therefore, the original marking 
could not have been optimal. After 
unmarking $e_i$, we again have 
$\sumrightu > \sumleftu$ and we can 
keep unmarking all edges on the 
right until the right side is fully 
unmarked and we have a strictly 
better marking than 
$\optimalmarking$ (see Figure 
\ref{fig7}).
Note that the other case 
($\sumrightu < \sumleftu$) can be 
handled symmetrically by fully 
unmarking the left side and 
repeatedly applying Eq. 
(\ref{delta2}).
\end{proof}
      \begin{figure*}[h]
    \centering
        \subfloat[An arbitrary marking with edges marked on both sides, such that $\sumrightu \geq \sumleftu$, as used in the proof of Lemma \ref{lemma11}.]{{\includegraphics[width=220pt]{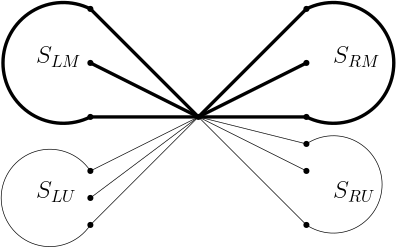}}}
        \subfloat[A strictly better marking than the one shown in (a), in which one side is fully unmarked. ]{{\includegraphics[width=220pt]{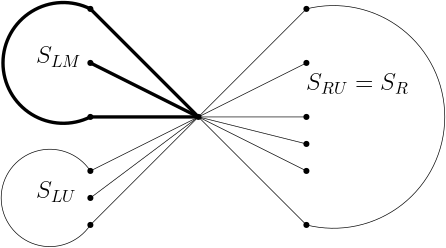}}}\

    \caption{The example tree used in the proof of Lemma \ref{lemma11}. (a) An arbitrary marking in which edges from both sides are marked, and more vertices are connected to the unmarked edges on the right ($\sumrightu \geq \sumleftu$). (b) A strictly better marking than (a) in which the heavier side is fully unmarked, as described in Lemma 
    \ref{lemma11}.}
    \label{fig7}
\end{figure*}
\subsubsection{Partial Markings}
\label{partial}
In Lemma \ref{lemma11}, we observed that no optimal marking has edges marked on both sides. In this section, we introduce the concept of partial markings, used to form optimal marking after merging a given edge $\merged$. A partial left (respectively right) marking, denoted by $\markingleft$ (respectively $\markingright$), is a marking with all edges on the right (respectively on the left) unmarked, and a subset of the edges on the left (respectively on the right) marked. We call a partial marking \textit{optimal} if its error count is less than any other partial marking for its respective side. Let $\markingleft^*$ and $\markingright^*$ denote the optimal partial left and right markings respectively, the following lemma is easy to prove.
\begin{lemma}
    After merging edge $\merged$ in a weighted tree with non-negative weights, the optimal marking $\optimalmarking$ is either $\markingleft^*$ or $\markingright^*$, depending on which one produces a smaller amount of error.
    \label{lemma12}
\end{lemma}
\begin{proof}
    Immediate from Lemma \ref{lemma11}.
\end{proof}
Applying the results of Lemma \ref{lemma12}, we can find the optimal marking $\optimalmarking$ by finding the optimal partial markings $\markingleft^*$ and $\markingright^*$, comparing their respective error values, and choosing the one with the smaller error value as the optimal marking. The question is how to find the optimal partial markings, and it is answered in the following lemma.

\begin{algorithm}
\caption{Graph Compression Algorithm for Trees}
\label{treesalg2}
\begin{algorithmic}[1]

\Procedure{TREE\_COMPRESSION\_SINGLE\_EDGE}{}    
\State \textbf{Input:} $T=(V,E)$ (A tree with $n$ vertices), an edge $\merged=(u,v)$ to be merged, the error function $\mathcal{E}(.)$
\State \textbf{Output:} A marking of edges $\optimalmarking$ with the optimal amount of error
\State Find all edges $E_L=\{(u,w)| (u,w) \in E, w \neq v\}$
\State $\leffs \xleftarrow[]{} |E_L|$, such that each edge $e_i$ in $E_L$ is connected to a subtree of size $L_i$ for all $i=\{1,\dots, \leffs\}$
\State Find all edges $E_R=\{(v,w)| (v,w) \in E, w \neq u\}$
\State $\riis \xleftarrow[]{} |E_R|$, such that each edge $e_i$ in $E_R$ is connected to a subtree of size $R_i$ for all $i=\{1,\dots, \riis\}$
\State $\sumleft \xleftarrow{} \sum_{\forall e_i \in E_L} L_i$
\State $\sumright \xleftarrow{} \sum_{\forall e_i \in E_R} R_i$
\State Remove $\merged$ from $T$ and merge its endpoints
\State $\markingleft^{*} \xleftarrow[]{} \emptyset$, $\markingright^{*}\xleftarrow[]{} \emptyset$
\For{ each $e_i \in E_L$}
 \If{$\sumleft -\sumright \leq L_i$}
\State $\markingleft^{*} \xleftarrow[]{} \markingleft^{*} \cup \{e_{i}\}$
\EndIf
\EndFor
\For{ each $e_i \in E_R$}
 \If{$\sumright -\sumleft \leq R_i$}
\State $\markingright^{*} \xleftarrow[]{} \markingright^{*} \cup \{e_{i}\}$
\EndIf
\EndFor
\State $\optimalmarking \xleftarrow[]{} \operatorname{argmin} (\mathcal{E}(\markingleft^{*}), \mathcal{E}(\markingright^{*}))$
\State \textbf{Return} $\optimalmarking$
\EndProcedure
\end{algorithmic}
\end{algorithm}

\begin{lemma}
The optimal partial marking $\markingleft^*$ consists of all edges $e_i$ (adjacent to $L_i$ vertices) such that
\begin{linenomath*}\begin{equation}
\sumleft -\sumright \leq L_i
    \label{construction1}
\end{equation}\end{linenomath*}
Similarly, the optimal partial marking $\markingright^*$ consists of all edges $e_i$ (adjacent to $R_i$ vertices) such that 
\begin{linenomath*}\begin{equation}
\sumright -\sumleft \leq R_i
    \label{construction2}
\end{equation}\end{linenomath*}
\label{lemma13}
\end{lemma}
\begin{proof}
    We first prove Eq. (\ref{construction1}) and derive Eq. (\ref{construction2}) by symmetry. Suppose we are trying to construct an optimal partial marking for the left side. We can do so by keeping the right side unmarked and marking edges on the left until the error can no longer be improved. Recall Eq. (\ref{delta3}) from the proof of Lemma \ref{lemma11}, we have to keep marking all edges $e_i$ (with cardinality $L_i$) until the error can no longer be improved, the change in the error value at each step is equal to:
    \begin{linenomath*}$$\Delta(\mle)= L_i \times \bigg(\sumleftm + (\sumleftu -L_i)+ \sumrightm - \sumrightu\bigg) $$\end{linenomath*}
    At each step, to get an improvement, we must have $\Delta(\mle) \leq 0 $: 
    \begin{linenomath*}$$(\sumleftm + (\sumleftu -L_i)+ \sumrightm - \sumrightu)\leq 0 \xrightarrow[]{\sumleftm+\sumleftu=\sumleft}\sumleft -L_i+ \sumrightm - \sumrightu\leq 0$$\end{linenomath*}
    However, since we are calculating a partial marking for the left side, we know by definition that the right side has to remain fully unmarked at all times, so we have $\sumrightu=\sumright$ and $\sumrightm=0$. Inserting these values in the above equation we get the following inequality for edges that \textit{improve} the partial left marking: 
    \begin{linenomath*}$$\sumleft -L_i - \sumright\leq 0\xrightarrow[]{}\sumleft - \sumright\leq L_i$$\end{linenomath*}
    Then, we can deduce that if an edge $e_i$ on the left satisfies $S_L-S_R \leq L_i$, it must be marked in $\markingleft^*$. Conversely, if an edge on the left $e_i$ is marked in $\markingleft^*$, it must satisfy $S_L-S_R \leq L_i$. To see why, assume $\markingleft^*$ includes an edge $e_i$ with $S_L-S_R > L_i$. Then, we can improve $\markingleft^*$ by unmarking $e_i$ (see Eq. (\ref{delta2})):
    \begin{linenomath*}$$\Delta(\umle)= L_i \times \bigg(- (\sumleftm - L_i) - \sumleftu -  \sumrightm + \sumrightu \bigg) =L_i\times \bigg( L_i -S_L+S_R\bigg) <0$$\end{linenomath*}
    which contradicts the optimality of $\markingleft^*$ and proves Eq. (\ref{construction1}). The other inequality (Eq. (\ref{construction2})) can be proven analogously by applying Eq. (\ref{delta4}).
\end{proof}
We present our linear-time algorithm for finding the optimal marking after merging an edge $e^*$ in Algorithm \ref{treesalg2}.
\begin{example}
    As an example, let us demonstrate how Algorithm \ref{treesalg2} finds the optimal marking for the tree of Figure \ref{fig6}. The optimal marking $\markingleft^*$ consists of all edges on the left, because:
\begin{itemize}
    \item $ 7-21\leq 2=L_1$
    \item $ 7-21\leq 2=L_2$
    \item $ 7-21\leq 3=L_3$
\end{itemize}
On the other hand, the optimal marking $\markingright^*$ consists of only one edge on the right, because: 
\begin{itemize}
    \item $21-7 \leq 20 = R_1$
    \item $21-7 > 1 =R_2$
\end{itemize}
Moreover, because $\markingright^*$ has a better error count than $\markingleft^*$, Algorithm \ref{treesalg2} returns $\markingright^*$ as the overall optimal marking $\optimalmarking$ which is the correct answer as depicted in Figure \ref{fig6}.
\end{example}
Now, we summarize our result in the following theorem.
\begin{theorem}
    Algorithm \ref{treesalg2} computes the optimal marking for a merged edge $\merged$ in $\mathcal{O}(|V|)$ time.
    \label{varyingthem}
\end{theorem}
\begin{proof}
    Immediate from Lemma \ref{lemma11}, Lemma \ref{lemma12}, and Lemma \ref{lemma13}.
\end{proof}

\subsection{Fractional Markings}
\label{fract}
In the previous section, we studied the marking problem under the assumption that each edge could either be fully marked or fully unmarked. In this section, we study a generalized version of the marking problem, called \textit{the fractional marking problem} (to be defined momentarily). We show that Algorithm \ref{treesalg2} does not err by assuming that each edge can either be fully marked or fully unmarked. 
\begin{definition}
With reference to a given merged edge $e^*$ in a graph $G=(V, E)$ with the associated weight function $w: E \Rightarrow \mathbb{R}_{\geq 0}$, and a new weight redistribution function $w^{\prime}: E \rightarrow \mathbb{R}_{\geq 0}$, an edge $e_i$ is said to be fractionally marked if $w^{\prime}\left(e_i\right)=w\left(e_i\right)+c_iw\left(e^*\right)$, $0<c_i<1$.
\label{fractional}
\end{definition}
Each neighbouring edge $e_i$ has thus an assigned $c_i$, which denotes the (possibly fractional) amount by which it is marked. We may sometimes refer to an edge $e_i$ with $c_i=1$ as a fully marked edge. An edge $e_i$ is marked by $\epsilon$ if its corresponding $c_i$ is set to $c^{\prime}_i=c_i + \epsilon$, and it is unmarked by $\epsilon$ if its corresponding $c_i$ is set to $c^{\prime}_i=c_i - \epsilon$.
        \begin{figure*}[h]
    \centering
        \subfloat[One edge from the left side and one from the right are fractionally marked.]{{\includegraphics[width=150pt]{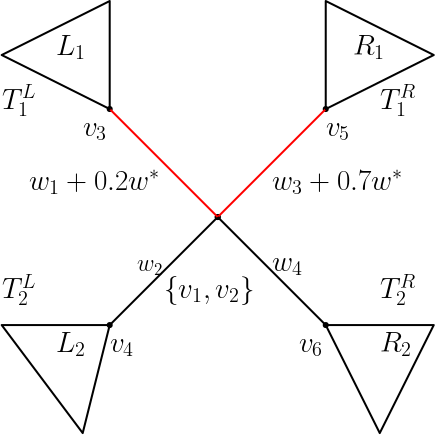}}}
        \hspace{2 cm}
        \subfloat[A succinct representation of (a) in which each fractionally marked edge $e_i$ is shown using its respective $c_i$ (Definition \ref{fractional}).]{{\includegraphics[width=150pt]{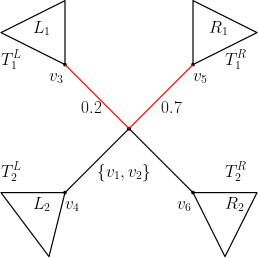}}}\

    \qquad
    \caption{An extension of Figure \ref{fig16} and Figure \ref{fig17} as an example of fractional markings. (a) One edge from the left side and one from the right are fractionally marked. (b) A succinct representation of (a) in which each fractionally marked edge $e_i$ is shown using its respective $c_i$ (Definition \ref{fractional}).  }
    \label{fig18}
\end{figure*}
\begin{definition}
\textbf{The Fractional Marking Problem for Weighted Trees:} Given a contracted edge $e^*$ in a weighted tree $T$ with non-negative weights, 
what subset of the neighbouring edges of $e^*$ should we fully mark or fractionally mark such that the error value of Eq. (\ref{hehe}) is minimized over all such possible subsets?
\label{fracdef}
\end{definition}
Similar to the previous section, we may omit some occurrences of $w^*$ from our calculation for convenience. We borrow our previous running example (Figure \ref{fig16} and Figure \ref{fig17}) and extend it to present an example of fractional markings in Figure \ref{fig18}. Figure \ref{fig18}-(a) depicts the tree of Figure \ref{fig16} with two edges fractionally marked. Figure \ref{fig18}-(b) illustrates a succinct representation of Figure \ref{fig18}-(a), where each fractionally marked edge $e_i$ is shown using its respective $c_i$ (Definition \ref{fractional}) and the weights of the unmarked edges are omitted. We use this succinct version often in the remainder of this section.

As a warm-up, we first present a property of any optimal marking that has at least one fractionally marked edge.
          \begin{figure*}[h]
    \centering
        \subfloat[]{{\includegraphics[width=150pt]{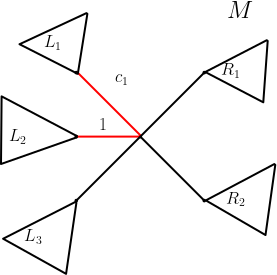}}}
        \hspace{2 cm}
        \subfloat[ ]{{\includegraphics[width=150pt]{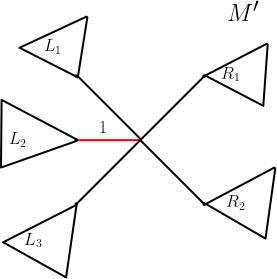}}}\

         \hspace{2 cm}\subfloat[]{{\includegraphics[width=150pt]{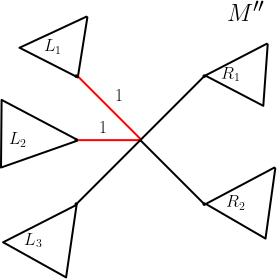}}}
                \qquad \hspace{2 cm}
    \qquad
    \caption{The example used in the proof of Lemma \ref{necessity}: (a) A hypothetical (for the sake of contradiction) optimal marking $M$ with a fractionally marked edge $e_1$ with marking value $c_1$. (b) Another marking $M^{\prime}$ resulting from unmarking $e_1$ in $M$. (c) A new marking $M^{\prime \prime}$ strictly better than $M$ ($\mathcal{E}(M^{\prime \prime }) < \mathcal{E}(M)$), yielding a contradiction and proving Lemma \ref{necessity}.}
    \label{fig19}
\end{figure*}
\begin{lemma}
Let $M$ be an optimal marking for a contracted edge $e^*$ (in a weighted tree) that has at least one fractionally marked edge $e^{\prime}$. Then, $M$ necessarily has marked edges on both sides.
\label{necessity}
\end{lemma}
\begin{proof}
    By contradiction. Suppose $M$ is an optimal marking with a 
    fractionally marked edge $e^{\prime}$, and suppose $M$ is a 
    partial left or right marking (Section \ref{partial}) with 
    marked edges only on the left or right respectively. Without 
    loss of generality, assume $M$ is a partial left marking 
    with a fractionally marked edge $e^{\prime}=e_1$. As 
    depicted in Figure \ref{fig19}-(a), assume $e_1$ has 
    cardinality $L_1$ and marking value $c_1$ (Definition 
    \ref{fractional}). We can obtain another marking 
    $M^{\prime}$ by unmarking $e_1$ (Figure \ref{fig19}-(b)). 
    Let $\mathcal{E}$ be the error function, then $\mathcal{E}
    (M)=\mathcal{E}(M^{\prime}) +\Delta_1 (\mle) $ and 
    $\mathcal{E}(M) <  \mathcal{E} (M^{\prime})$ because $M$ is 
    an optimal marking.  Therefore, $\Delta _1(\mle) < 0$ when 
    marking $e_1$ back in $M^{\prime}$. We now formulate 
    $\Delta_1(\mle)$ when marking $e_1$ in $M^{\prime}$ by 
    $c_1$. 
    
    \begin{linenomath*}$$\Delta_{1}(\mle)=c_1 \times (X)$$\end{linenomath*}
    where $X= L_1 \times \bigg(\sumleftm + (\sumleftu -L_1)+ 
    \sumrightm - \sumrightu\bigg)= L_1 \times \bigg(\sumleftm + 
    (\sumleftu -L_1) - S_R\bigg)$ because $\sumrightm=0$ and 
    $S_R= \sumrightu$. However, because $\Delta_1 (\mle)<0$, we 
    have that $X<0$ and we can fully mark $e_1$ in $M^{\prime}$ 
    to get another marking $M^{\prime \prime}$ (Figure 
    \ref{fig19}-(c)). The amount of error change of this mark 
    operation is equal to: 
    
    \begin{linenomath*}$$\Delta_{2}(\mle)= (c_1) \times X + (1-c_1) \times X$$\end{linenomath*}
    Noting that $\mathcal{E}(M)=\mathcal{E}(M^{\prime}) 
    +\Delta_1 (\mle) $ we have: 
    
    \begin{linenomath*}$$\mathcal{E}(M^{\prime \prime})= \mathcal{E}(M^\prime) + 
    \Delta_2(\mle) =\underbrace{\mathcal{E}(M^{\prime})+(c_1) 
    \times X}_{=\mathcal{E}(M^{\prime}) 
    +\Delta_1(\mle)=\mathcal{E}(M)} + (1-c_1) \times X< 
    \mathcal{E}(M)$$\end{linenomath*}
    Therefore $M^{\prime \prime }$ is a strictly better marking 
    than $M$, contradicting our assumption that $M$ is an 
    optimal marking.
\end{proof}

Lemma \ref{necessity} states that an optimal marking with fractionally marked edges cannot be a partial left or right marking (as defined in Section \ref{partial}). 

We now present the main result of this section. 
\begin{lemma}
    Let $M$ be an optimal marking for a contracted edge $e^*$ (in a weighted tree) that contains both fractionally and fully marked edges. Then, $M$ can be transformed into another optimal marking $M^{\prime}$ that contains no fractionally marked edges. 
    \label{cons1}
\end{lemma}
\begin{proof}
    We assume $M$ is an optimal marking that contains fractionally marked edges. We consider several possible cases and for each case, we present a transformation technique that does not worsen the marking with reference to the error function (Eq. (\ref{hehe})). The repeated application of these transformations converts $M$ into another marking $M^{\prime}$ with no fractionally marked edges.

    Let $M$ be an optimal marking that contains at least one fractionally marked edge. Using Lemma \ref{necessity}, we may assume $M$ contains marked edges on both sides. Throughout this proof, we use $0<c_i\leq 1$ to denote the marking value for an edge $e_i$, such that $e_i$ is fractionally marked if $0<c_i< 1$ (see Definition \ref{fractional}).
    \begin{itemize}
        \item \textbf{Case 1:} $M$ contains two marked edges $e_1$ (with cardinality $L_1$) and $e_2$ (with cardinality $R_1$) on the left and right respectively such that $c_1 + c_2 > 1$ (Figure \ref{fig20}-(a)).
        \end{itemize}
        Let $c_1 + c_2 =1+ \epsilon$ for $\epsilon >0$. We show 
        that fractionally unmarking either $e_1$ or $e_2$ by 
        $\epsilon$ does not worsen the error, and this change 
        transforms $M$ into another optimal marking $M^{\prime}$ 
        in which $c^{\prime}_1+c^{\prime}_2=1$. 
        We generalize the proof of Lemma \ref{lemma11} and 
        define $\sumleft =\sum_{i=1}^{\leffs}L_i$ and $\sumright 
        =\sum_{i=1}^{\riis}R_i$ as the total sum of all edge 
        cardinalities on the left and right sides respectively. 
        Let $\sumleft^{\prime}= \sumleft- L_1$ and 
        $\sumright^{\prime}= \sumright- R_1$ and without loss of 
        generality, assume $\sumright^{\prime} \geq 
        \sumleft^{\prime}$. We unmark $e_2$ by $\epsilon$, 
        setting $c^{\prime}_2=c_2 -\epsilon$ (Figure \ref{fig20}-
        (b)). Note that we must necessarily have $c_2 \geq 
        \epsilon$ because otherwise $c_1+ c_2 < 1 + \epsilon= 
        c_1+c_2$, which is a contradiction. We now show this 
        operation does not worsen the marking. Between the 
        vertices of $e_2$ and the vertices of all other edges on 
        the right side, the error is reduced by $-\epsilon 
        \times w^*$. Now, let $e_j \neq e_1$ be some edge on the 
        left. The error between the vertices of $e_2$ and $e_j$ 
        is increased by at most $\epsilon \times w^*$. If $e_j$ 
        has marking value $c_j$, then this operation may 
        increase the error between the vertices of $e_2$ and 
        $e_j$ by $|w^*- (c_2 \times w^* +c_j \times w^* -
        \epsilon \times w^*)| -|w^*- (c_2 \times w^* +c_j \times 
        w^* )| \leq \epsilon \times w^*$ using Corollary 
        \ref{obs}.
        Therefore, this unmark operation changes the error by: 
                \begin{linenomath*}$$
\begin{array}{cc}
    \Delta_1(\unmr) \leq R_1 \times \epsilon \times w^* \times \bigg( - \sumright^{\prime} \underbrace{-  L_1}_{<0} + \sumleft^{\prime}\bigg)  & <  
   R_1 \times \epsilon \times w^* \times (\underbrace{-\sumright^{\prime} +\sumleft^{\prime}}_{\leq 0} )\leq  0
\end{array}
$$\end{linenomath*}
      \begin{figure*}[h]
    \centering
        \subfloat[]{{\includegraphics[width=200pt]{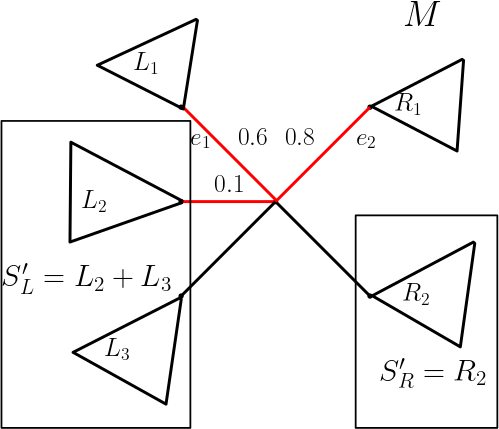}}}
        \hspace{2 cm}
        \subfloat[]{{\includegraphics[width=200pt]{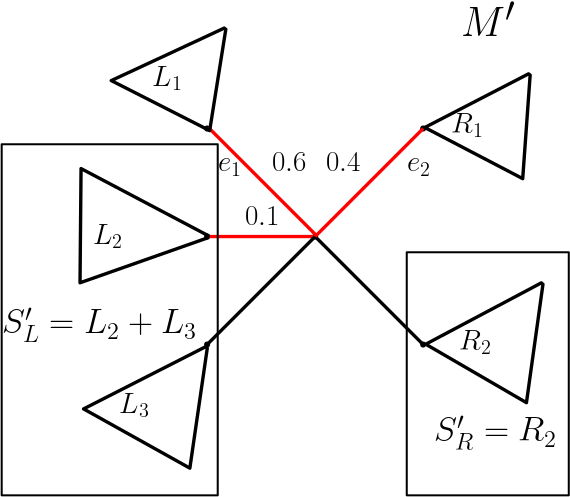}}}\

    \qquad
    \caption{Case 1 in the proof of Lemma \ref{cons1}: (a) There exist two edges $e_1$ and $e_2$ such that $c_1+ c_2 =0.6 + 0.8= 1.4 > 1$, $S_R^{\prime}= R_2 \geq S_L^{\prime}= L_2+L_3$. (b) Another marking $M^{\prime}$ where $e_2$ is unmarked by $\epsilon=0.4$ such that $\mathcal{E}(M^{\prime}) \leq \mathcal{E}(M)$.}
    \label{fig20}
\end{figure*}

      \begin{figure*}[h]
    \centering
        \subfloat[]{{\includegraphics[width=150pt]{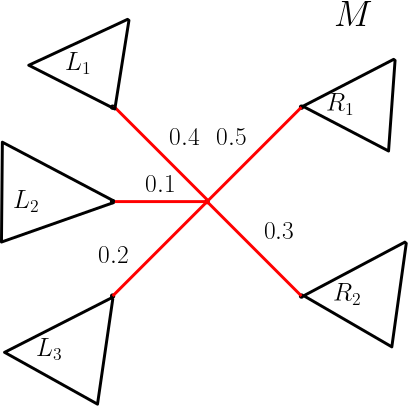}}}
        \hspace{2 cm}
        \subfloat[]{{\includegraphics[width=150pt]{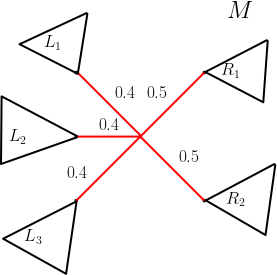}}}\

    \qquad
    \caption{Case 2 in the proof of Lemma \ref{cons1}: (a) Case 2-1: For all edges $e_1, e_2$ on opposite sides $c_1+c_2 \leq 1$, and there exist two edges $e_i, e_j$ on one side with $c_i \neq c_j$ (for instance $0.4$ and $0.1$ on the left) (b) Case 2-2: For all edges $e_1, e_2$ on opposite sides $c_1+c_2 \leq 1$, for all edges $e_i$ on the left $c_i=\epsilon_1=0.4$, and for all edges $e_j$ on the right $c_j=\epsilon_2=0.5$.}
    \label{fig21}
\end{figure*}
\begin{itemize}
            \item  \textbf{Case 2:} For all pairs of marked edges $e_1$ and $e_2$ on the left and right respectively $c_1 +c_2 \leq 1$ (Figure \ref{fig21}).
\end{itemize}
We consider two subcases: 
\begin{itemize}
    \item \textbf{Case 2-1:} There exist two edges $e_i$ and $e_j$ on one side (left or right) such that $c_i \neq c_j$ (Figure \ref{fig21}-(a)). 
\end{itemize}
Figure \ref{fig21}-(a) is an example of such marking $M$ in 
which $c_1+ c_2 \leq 1 $ for any two edges $e_1$ and $e_2$ on 
opposite sides and at least two edges $e_i$ and $e_j$ on the 
left with $c_i \neq c_j$. Without loss of generality, we may 
assume $c_i < c_j$. Then, we can set $c^{\prime}_i= c_j$ without 
increasing the error. Due to the properties of this subcase, we 
may get another marking $M^{\prime}$ by unmarking $e_i$ (setting 
$c^{\prime}_i=0$) and then marking it by $c_j$ to get a third 
marking $M^{\prime \prime }$ with $\mathcal{E}(M^{\prime \prime 
}) \leq \mathcal{E} (M)$. Similar to the proof of Lemma 
\ref{necessity}, we have: 
        \begin{linenomath*}$$\mathcal{E}(M^{\prime \prime})=\underbrace{\mathcal{E}
(M^{\prime}) + c_1 \times X }_{=\mathcal{E}(M)}+ (c_2- c_1) 
\times X$$ \end{linenomath*}
\begin{itemize}
    \item \textbf{Case 2-2:} For all marked edges $e_i$ (with 
    $c_i >0$) on the left $c_i=\epsilon_1$, and for all marked 
    edges $e_j$ (with $c_j>0$) on the right $c_j =\epsilon_2$ 
    ($\epsilon_1 + \epsilon_2 \leq 1$) (Figure \ref{fig21}-(b)). 
\end{itemize}
For this case, we simply show that the error associated with the 
optimal partial marking (Lemma \ref{lemma13}) is a lower bound 
on $\mathcal{E}(M)$, or $\operatorname{min}(\mathcal{E}
(\markingright^*), \mathcal{E}(\markingleft^*))\leq \mathcal{E}
(M)$. Without loss of generality, we assume that 
$\operatorname{min}(\mathcal{E}(\markingright^*), \mathcal{E}
(\markingleft^*))= \mathcal{E}(\markingleft^*)$. \\ 
Let $E_L$ and $E_R$ be the set of marked edges in 
$\markingleft^*$ and $\markingright^*$ respectively. From Lemma 
\ref{lemma13}, we know that for each $e_i \in E_L$, $S_L-
S_R \leq L_i$ and for each $e_i \in E_R$, $S_R-S_L \leq 
R_i$. We show that the set of marked edges in $M$ is precisely 
equal to $E_L \cup E_R$. Let $e_i$ be any marked (with reference 
to $M$) edge on the left, and let $e_j$ be any marked edge on 
the right. Because $c_i +c_j \leq 1$, unmarking $e_i$ by 
$\epsilon\leq c_i$ increases the error between the vertices of 
$e_i$ and $e_j$ by $\epsilon \times w^*$:

        \begin{linenomath*}$$\underbrace{|w^*-(c_i\times w^* + c_j\times w^* -\epsilon 
\times w^*)|}_{=w^*-(c_i \times w^*+ c_j \times w^* -\epsilon 
\times w^*) \text{ because }c_i+c_j \leq 1}- \underbrace{|w^*-
(c_i\times w^* + c_j \times w^*)|}_{=w^*-(c_i \times w^*+ c_j 
\times w^*) \text{ because }c_i+c_j \leq 1} = \epsilon \times 
w^*$$\end{linenomath*}

Furthermore, unmarking $e_i$ by $\epsilon$ decreases the error 
between the vertices of $e_i$ and the vertices of all other 
edges on the left by $-\epsilon \times w^*$. Therefore, 
unmarking $e_i$ by $\epsilon$ changes $\mathcal{E}(M)$ by: 

        \begin{linenomath*}$$\Delta_1 (\umle)= L_i \times \epsilon \times w^* \times (-(S_L 
- L_i) + S_R )=L_i \times \epsilon \times w^* \times (-S_L + L_i 
+ S_R )$$        \end{linenomath*}

Because $M$ is an optimal marking, $\Delta_1 (\umle) \geq 0$ 
and: 
        \begin{linenomath*}$$-S_L + L_i + S_R \geq  0 \xrightarrow{}L_i \geq  S_L - S_R$$        \end{linenomath*}
Conversely, we may assume that any edge $e_i$ on the left satisfying $L_i \geq  S_L - S_R$ is marked in $M$; because 
otherwise, we could improve $M$ by marking $e_i$\footnote{The 
proof for this claim is almost identical to the one provided in 
the proof of Lemma \ref{lemma13}. Here, we omitted the details 
to avoid repetition.}. Similar reasoning can be applied to any 
marked edge $e_i$ on the right.
We now conclude the proof. First, note that:
        \begin{linenomath*}$$\mathcal{E}(\markingleft^*)= \underbrace{\mathcal{E}
(M_0)}_{\text{The error associated with the empty marking}} 
+\underbrace{\sum_{e_i \in E_L} L_i  \times w^* \times (S_L-L_i-
S_R)}_{\text{The sum of all } \Delta(\mle) \text{'s that 
transform }M_0 \text{ into }\markingleft^*}$$        \end{linenomath*}
and 
        \begin{linenomath*}$$\mathcal{E}(\markingright^*)= \underbrace{\mathcal{E}
(M_0)}_{\text{The error associated with the empty marking}} 
+\underbrace{\sum_{e_i \in E_R} R_i  \times w^* \times (S_R-R_i-
S_L)}_{\text{The sum of all } \Delta(\mr) \text{'s that 
transform }M_0 \text{ into }\markingright^*}$$        \end{linenomath*}

where $M_0$ is a trivial marking with no marked edges. 

On the other hand, $M$ can be constructed by first marking all 
edges $e_i$ in $E_L$ by $\epsilon_1$ and then marking all edges in $E_R$ by $\epsilon_2$.
        \begin{linenomath*}
\begin{align*}
\mathcal{E}(M)=& \underbrace{\mathcal{E}(M_0)}_{\text{The error associated with the empty marking}} +\underbrace{\epsilon_1 
\times \sum_{e_i \in E_L} L_i  \times w^* \times (S_L-L_i-
S_R)}_{\text{The sum of all } \Delta(\mle) \text{'s by 
}\epsilon_1 }\\
+& \underbrace{\epsilon_2 \times \sum_{e_i \in E_R} R_i\times w^* \times (S_R-R_i-S_L)}_{\text{The sum of all }\Delta(\mr) \text{'s by }\epsilon_2 }
\end{align*} \end{linenomath*}
From our assumption, $\mathcal{E}(\markingleft^*) \leq 
\mathcal{E}(\markingright^*)$. We have $S_L-L_i-S_R \leq 0 $ for 
all $e_i \in E_L$ and $S_R-R_i-S_L \leq 0$ for all $e_i \in 
E_R$. We get: 
        \begin{linenomath*}
$$\mathcal{E}(\markingleft^*) \leq \mathcal{E}
(\markingright^*)\xrightarrow[]{} \sum_{e_i \in E_L}L_i\times w^* \times (S_L-L_i-S_R) \leq \sum_{e_i \in E_R} R_i \times w^* \times (S_R-R_i-S_L)$$        \end{linenomath*}

On the other hand: 

\begin{align}
\mathcal{E}(M)&= \mathcal{E}(M_0) + \epsilon_1 \times \sum_{e_i 
\in E_L} L_i  \times w^* \times (S_L-L_i-S_R)+ \epsilon_2 \times 
\sum_{e_i \in E_R} R_i  \times w^* \times (S_R-R_i-S_L)
\nonumber \\
&\geq 
\mathcal{E}(M_0) + \epsilon_1 \times \sum_{e_i \in E_L} L_i\times w^* \times (S_L-L_i-S_R)+ \epsilon_2 \times \sum_{e_i \in 
E_L} L_i  \times w^* \times (S_L-L_i-S_R) \nonumber \\
 &=  \mathcal{E}(M_0) + (\epsilon_1 +\epsilon_2) \times 
 \sum_{e_i \in E_L} L_i  \times w^* \times (S_L-L_i-S_R) 
 \xrightarrow[S_L -L_i -S_R \leq 0]{\epsilon_1 +\epsilon_2\leq 
 1} \nonumber \\
 &\geq
\mathcal{E}(M_0) +  \sum_{e_i \in E_L} L_i  \times w^* \times (S_L-L_i-S_R) =\mathcal{E}(\markingleft^*) \nonumber \\
&\xrightarrow[]{} \mathcal{E}(M) \geq \mathcal{E}(\markingleft^*) 
\end{align}

Thus, $\operatorname{min}(\mathcal{E}(\markingleft^*), 
\mathcal{E}(\markingright^*))$ is a lower bound on the error of 
any such marking. \\

We now conclude the proof by stating that any optimal marking 
with fractionally marked edges can be transformed into another 
optimal marking with no fractionally marked edges.  Let $M$ be 
any such marking. If $M$ satisfies the conditions of Case 1, we 
repeatedly apply the transformation method of Case 1 until it 
satisfies the conditions of Case 2-1. We then repeatedly apply 
the construction method of Case 2-1 until $M$ satisfies the 
conditions of Case 2-2. Finally, if $M$ satisfies the conditions 
of Case 2-2, we have already shown that $\mathcal{E}(M)$ is 
lower bounded by the optimal partial marking of Lemma 
\ref{lemma13} which has no fractionally marked edges.

\end{proof}
\section{Conclusion and Open Problems}
In this paper, we studied the problem of distance-preserving 
graph compression for weighted paths and trees. We first 
presented a brief literature review of some related work in this 
domain, noting that one particular aspect of the problem is 
understudied. More specifically, there has been little attention 
in the literature to the problem of optimally compressing a 
given set of edges. To address this, we presented optimal 
algorithms for compressing any set of $k$ edges in a weighted 
path and for optimally compressing a single edge in a weighted 
tree. We tackled the problems in an incremental order of 
difficulty. For weighted paths, we first solved the problem of 
optimally compressing a single edge, then we generalized it to 
any set of $k$ independent edges. Finally, we provided an 
optimal approach to compressing any contiguous subset of edges 
in a weighted path.\\
We then generalized our scope to weighted trees, where we 
studied the problem of optimally compressing a single edge. To 
this end, we first studied the easier case in which the subtrees 
of both sides of the merged edge had equal sizes. Finally, we 
generalized our results to the case in which different subtrees 
were of different sizes.\\ 
We now note some potential avenues of future studies:
\begin{openproblem}
    Can we solve the distance-preserving graph compression problem for general graphs in polynomial time?
\end{openproblem}
The above problem would indeed be a natural extension of this 
paper. The complexity of the weight redistribution problem for 
general graphs is still unknown. However, it appears that the 
related problem of finding the contracted edges is unlikely to 
be solved in polynomial time. Bernstein \textit{et al.} 
\cite{bernstein} showed that $\operatorname{CONTRACTION}$ 
(defined in Section 1) is NP-hard even if the underlying graph 
is just a weighted cycle. In a graph with cycles, some vertices 
are connected via multiple paths. Therefore, after merging a 
single edge, several shortest paths that traverse that edge may 
need to be rerouted using completely different edges, making the 
analysis much more difficult.
\begin{openproblem}
How could we find an optimal redistribution strategy that also 
minimizes the error between all pairs of vertices in $V_m$?
\end{openproblem}
Note that even if some weight redistribution minimized the error 
between two nodes in different supernodes, it would still be non-trivial to do the same for two vertices that are placed in a 
single supernode. Obviously, a trivial solution would be to 
store the shortest path weights between the vertices in one 
supernode as separate table entries. However, such an approach 
would defeat the whole purpose of graph compression, which is to 
reduce memory requirements. 
\begin{openproblem}
    For the optimal weight redistribution problem, are there any 
    better cost models (error functions)?
\end{openproblem}
As stated in Section \ref{prel}, in this paper we defined the 
error function as the sum of
the absolute differences of the shortest path lengths between different pairs of nodes before and after redistributing the 
weights. However, exploring alternative cost functions that 
better capture the distance-based similarity between a modified 
graph and its original version can open up exciting research 
avenues. Investigating whether there exist other cost functions 
that provide a more accurate measure of closeness between graphs 
can lead to valuable research opportunities.
\label{conclusion}
\bibliographystyle{plain}
\bibliography{main.bib}
\end{document}